\documentclass[11pt]{article}
\usepackage[margin=1in]{geometry}

\usepackage{graphicx}
\usepackage{amsfonts}
\usepackage{amsmath}
\usepackage{amssymb}
\usepackage{amsthm}
\usepackage{thmtools}
\usepackage{url}

\usepackage[breakable]{tcolorbox}

\usepackage{algorithm,algpseudocode}

\usepackage{dsfont}
\usepackage{fancyhdr}
\usepackage{indentfirst}
\usepackage{enumerate}
\usepackage[normalem]{ulem}
\usepackage{mathrsfs}
\usepackage{multirow}
\usepackage{diagbox}

\usepackage{color}

\definecolor{DarkGreen}{rgb}{0.2,0.6,0.2}

\definecolor{purple}{rgb}{0.6,0.3,0.8}

\usepackage{natbib}
\usepackage{comment}
\usepackage[colorlinks=true,citecolor=blue,hypertexnames=false]{hyperref}

\def\d{\mathrm{d}}

\newcommand{\E}{\mathbb{E}}

\newcommand{\R}{\mathbb{R}}

\newcommand{\N}{\mathbb{N}}
\newcommand{\p}{\mathbb{P}}

\newcommand{\Q}{\mathbb{Q}}

\renewcommand{\ge}{\geqslant}
\renewcommand{\le}{\leqslant}
\renewcommand{\geq}{\geqslant}
\renewcommand{\leq}{\leqslant}
\renewcommand{\epsilon}{\varepsilon}

\renewcommand{\cdots}{\dots}

\usepackage[capitalise,noabbrev]{cleveref}

\theoremstyle{plain}
\newtheorem{theorem}{Theorem}[section]
\newtheorem*{theorem*}{Theorem}
\newtheorem{corollary}[theorem]{Corollary}
\newtheorem{fact}[theorem]{Fact}
\newtheorem{lemma}[theorem]{Lemma}
\newtheorem{proposition}[theorem]{Proposition}
\newtheorem{conj}{Conjecture}
\theoremstyle{definition}
\newtheorem{definition}[theorem]{Definition}

\newtheorem{example}[theorem]{Example}

\newcommand{\width}{w}

\theoremstyle{remark}
\newtheorem{remark}{Remark}

\usepackage[onehalfspacing]{setspace}

\newcommand{\SP}{\mathrm{SP}}
\newcommand{\ind}{\mathbf 1}

\usepackage{autonum}

\renewcommand{\P}{\mathbb{P}}

\newcommand{\Fcal}{\mathcal F}
\newcommand{\Pcal}{\mathcal{P}}
\newcommand{\Pin}{{\P \in \Pcal}}
\newcommand{\Var}{\mathrm{Var}}
\newcommand{\widebar}[1]{\overline{#1}}
\newcommand{\KL}{{\mathrm{KL}}}
\newcommand{\KLinf}{{\mathrm{KL}_\mathrm{inf}}}

\newcommand{\brackell}{{(\ell)}}
\newcommand{\bracku}{{(u)}}

\newcommand{\bx}{\mathbf x}
\newcommand{\bE}{\mathbf E}

\newcommand{\NN}{\mathbb N}
\newcommand{\EE}{\mathbb E}

\newcommand{\1}{\mathbf{1}}

\def\ddefloop#1{
  \ifx
    \ddefloop#1
  \else\ddef{#1}
    \expandafter
    \ddefloop
  \fi}
\def\ddef#1{\expandafter\def\csname #1cal\endcsname{\ensuremath{\mathcal{#1}}}}

\ddefloop
ABCDEFGHIJKLMNOPQRSTUVWXYZ
\ddefloop

\def\ddefloop#1{
  \ifx
    \ddefloop#1
  \else\ddef{#1}
    \expandafter
    \ddefloop
  \fi}
\def\ddef#1{\expandafter\def\csname #1sf\endcsname{\ensuremath{\mathsf{#1}}}}

\ddefloop
ABCDEFGHIJKLMNOPQRSTUVWXYZ
\ddefloop

\def\ddefloop#1{
  \ifx
    \ddefloop#1
  \else\ddef{#1}
    \expandafter
    \ddefloop
  \fi}
\def\ddef#1{\expandafter\def\csname #1#1\endcsname{\ensuremath{\mathbb{#1}}}}

\ddefloop
ABCDEFGHIJKLMNOPQRSTUVWXYZ
\ddefloop

\usepackage{footmisc}
\makeatletter
\let\old@fnsymbol\@fnsymbol
\renewcommand{\@fnsymbol}[1]{
  \ifnum#1=5 \ensuremath{\diamond}
  \else
    \old@fnsymbol{#1}
  \fi
}
\makeatother

\begin{document}

\title{Combining e-values using demi-supermartingales}

\author{
Jiahao Ming\thanks{Department of Statistics and Actuarial Science, University of Waterloo, Canada.
\texttt{j5ming@uwaterloo.ca}}
\and 
Aaditya Ramdas\thanks{Department of Statistics and Data Science, Carnegie Mellon University, USA.
\texttt{aramdas@cmu.edu}}
\and 
Yi Shen\thanks{Department of Statistics and Actuarial Science, University of Waterloo, Canada. \texttt{yi.shen@uwaterloo.ca}}
\and 
Ruodu Wang\thanks{Department of Statistics and Actuarial Science, University of Waterloo, Canada.   \texttt{wang@uwaterloo.ca}}
\and Ian Waudby-Smith\thanks{Department of Statistics, University of California, Berkeley,  USA. \texttt{ianws@berkeley.edu}} 
}

\maketitle
\begin{abstract}
We present a new method for combining e-variables through demi-supermartingales, which settles an old conjecture in the literature on nonparametric mean testing. It also provides an explicit concentration bound for a certain Kullback--Leibler-type statistic arising in the stochastic multi-armed bandit literature. All of these combination results hold for independent e-variables as well as for the class of co-valid e-variables, whose dependence structure lies somewhere between independence and sequential validity.
The results are further generalized to compound e-variables. 
The proofs proceed by analyzing elementary symmetric polynomials and their behavior as nonnegative demi-supermartingales. 


\end{abstract}

\section{Introduction}

E-values bridge stochastic processes   and statistical testing, offering several advantages over p-values in sequential testing, multiple testing, and post-hoc-$\alpha$ decision-making.
 A comprehensive statistical treatment of the topic is given by \cite{RW25}.  For a given (possibly composite) null hypothesis, an e-variable  is a nonnegative (extended) random variable with mean at most $1$ under the null hypothesis.  
An e-variable or its realization is called an e-value. 

Let us begin with the following question from \cite{VW21} (its background will be explained later in the paper): given independent e-variables $E_1, \dots, E_n$, how do we combine them into a p-value or a test?
A standard approach for this purpose is described below.
Define, for $\lambda \in [0,1]$, the process $M(\lambda)$ by  $M_0(\lambda)=1$ and 
  \begin{equation}
      \label{eq:M}  M_n (\lambda) =\prod_{i=1}^n \left( (1-\lambda ) + \lambda  E_i\right) \mbox{~~~for $n \in \N$}.
  \end{equation} 
  It is straightforward to see that $M(\lambda)$ is a supermartingale, and 
  using an inequality of~\cite{ville1939etude}, one gets for any fixed $\lambda\in [0,1]$, \begin{equation}\label{eq:ville} \p\left (\sup_{n\ge 1} M_n(\lambda)\ge \frac 1\alpha \right)\le \alpha \mbox{~~~for $\alpha \in (0,1)$.}\end{equation}
  Indeed, \eqref{eq:ville} exemplifies a fundamental technique for sequential hypothesis testing with e-values.\footnote{More generally, $\lambda$ is allowed to vary across different values of $n$ as long as it forms a predictable process; see, e.g., \cite{WR24}.} Said differently, the reciprocal of $\sup_{n\ge 1} M_n(\lambda)$ is a p-value.

 In this paper, we look at  $M(\lambda)$ from a very different angle; we take the supremum over $\lambda$ rather than over $n$. \citet[Theorem 4.1]{WZ03} showed that this quantity coincides with the nonparametric likelihood ratio for iid nonnegative data with the null hypothesis that its mean is no more than $1$, and made the following conjecture. \begin{conj}
 \label{conjecture:gwz}

 Let $M_n(\lambda)$ be given in  \eqref{eq:M}.  \begin{enumerate}[(i)]
     \item  {\rm (Wang--Zhao).}
          For iid e-variables $E_1,\dots, E_n$, \begin{equation}\label{eq:gwz-conjecture-evalue}
        \P \left ( \sup_{\lambda \in [0, 1]} M_n(\lambda) \geq x \right ) \leq \frac{1}{x} \quad \mbox{for $x>0$}.
    \end{equation}
     \item {\rm (Gaffke).}   For independent   e-variables $E_1,\dots, E_n$, \eqref{eq:gwz-conjecture-evalue} holds.
 \end{enumerate}

\end{conj}
  \citet{WZ03} proved their version of \cref{conjecture:gwz} in the case where $n \in \{1, 2\}$ and left as an open problem whether the result holds for arbitrary $n \in \NN \setminus \{1,2\}$.  The \cite{gaffke2005three}  version of  \cref{conjecture:gwz} is stronger.   The above authors did not use the term ``e-variable'' in their work, but  \cref{conjecture:gwz} is   a mathematically equivalent re-statement of their explicit conjectures.

We obtain a surprisingly strong result: \cref{theorem:maximal-inequality-co-valid} settles  both versions of the conjecture in the positive, but with  a statistic that is larger than $\sup_{\lambda \in [0, 1]} M_n(\lambda)$ and under a dependence condition more general than independence;
we call e-variables satisfying it \emph{co-valid}.
Co-valid e-variables are special cases of sequential e-variables, which suffice for~\eqref{eq:ville}. 
The quantity $\sup_{\lambda \in [0, 1]} M_n(\lambda)$ represents the best outcome from choosing a constant betting strategy after observing all the data,
and therefore it is  potentially more powerful than \eqref{eq:ville} in some situations.
Moreover, in \cref{thm:average-mean-sympol}, we obtain a further generalization: the inequality \eqref{eq:gwz-conjecture-evalue} in \cref{conjecture:gwz} 
holds for nonnegative random variables that have average mean at most $1$; these random variables are called
compound e-values by \cite{ignatiadis2024asymptotic} when  testing a global null.

The logarithm of $\sup_{\lambda \in [0,1]} M_n(\lambda)$ has a dual form that results in it being called the KL-inf statistic in the multi-armed bandit literature~\citep{honda2010asymptotically,agrawal2021optimal}. These works use deterministic regret bounds to establish time-uniform concentration bounds for KL-inf that correct the $1/\alpha$ term in~\eqref{eq:ville} by factors that are \emph{polynomial in $n$}. Our result can thus be seen as providing \emph{regret-free concentration} for KL-inf at a fixed sample size $n$. 

The proof strategy that we develop throughout this paper is distinct from those typically encountered in the literature on e-variables. When e-variables are independent (or more generally, sequential), an (anytime-valid) p-value can be obtained by taking the product of those e-variables, recognizing that its cumulative product forms a nonnegative supermartingale, and applying \citeauthor{ville1939etude}'s [\citeyear{ville1939etude}] inequality. Indeed, this approach is now ubiquitous throughout the literature; see \citet[Section 7]{RW25}. However, it is not clear how one would employ such a technique in attempting to prove \cref{conjecture:gwz} for instance, because $ \lambda\mapsto \prod_{i=1}^n \left ( 1 - \lambda  +\lambda E_i \right )$  is not a (backward or forward) supermartingale. As we demonstrate in the proof of \cref{theorem:maximal-inequality-co-valid}, the random variables over which the maximum is taken nevertheless form a \emph{demi-supermartingale}. It is natural to wonder whether an analogue of Ville's inequality also holds for nonnegative demi-supermartingales, and in \cref{section:demi-ville}, we show that this is in fact the case, leading to a strict generalization of \citet{ville1939etude}.

The remainder of the paper is organized as follows. 
After establishing  Ville's inequality for nonnegative demi-supermartingales in
\Cref{section:demi-ville}, 
\Cref{section:co-validity} introduces co-valid e-variables and proves \cref{theorem:maximal-inequality-co-valid}, which settles the Wang--Zhao conjecture under this more general dependence condition. 
\Cref{section:bdd-means} applies the main result to confidence intervals for means of bounded random variables, derives regret-free concentration for the KL-inf statistic, studies the asymptotic widths of the resulting intervals, and establishes a more general KL-inf duality theorem. 
\Cref{sec:heterogeneous-means} extends the main results to the setting of heterogeneous means and compound e-values, with \cref{thm:average-mean-sympol} being the stronggest result in the paper. 
\Cref{sec:comp} presents a fast algorithm for computing the ``SymPol" statistic constructed from our theoretical results. 
The final section concludes, and some omitted proofs are collected in the appendix.

\section{Strengthening Ville's inequality with demi-supermartingales}
\label{section:demi-ville}

In this section, we recap the concepts of demi-martingales and demi-supermartingales \citep{rao2011associated}, culminating in a generalization of Ville's inequality for demi-supermartingales. These are then used in the proof of \cref{theorem:maximal-inequality-co-valid} and lead to new hypothesis tests with combined e-variables. Indeed, as far as we are aware,
this paper is the first to connect e-variables and p-values with demi-supermartingales. We hope that this leads to new ways of testing via the construction of new demi-supermartingales. 

For the definitions and results to follow, let $T \in \NN \cup \{\infty\}$ and let $\Tcal = \{0,1, \dots, T\}$   denote a ``time horizon'' where we use the convention that if $T = \infty$, then $\Tcal = \NN\cup \{0 \}$.

\begin{definition}[Demi-(super)martingales]
\label{def:demi}
An integrable process $(M_k)_{k\in \Tcal}$ is said to be a \emph{demi-martingale} (resp.~\emph{demi-supermartingale}) if 
\begin{equation}\label{eq:demi-def}
\E[(M_{k}-M_{k-1})g(M_0,\dots,M_{k-1})] \leq 0\quad\text{for all } \quad k \in \Tcal \setminus \{0\},
\end{equation}
and for every decreasing (resp.~decreasing and nonnegative) $g$ for which the expectation is defined.
\end{definition}

It is straightforward to check that all martingales are demi-martingales, all supermartingales are demi-supermartingales, and all demi-martingales are demi-supermartingales. Let us now derive an inequality which will serve as a drop-in replacement for Doob's optional stopping theorem in the proof of Ville's inequality for demi-supermartingales.

In what follows, we use the following convention for $M_\infty$: if $ T<\infty$, then $M_\infty=M_T$; 
if $T=\infty$, then 
$M_\infty$ is the a.s.~limit of $M_k$ as $k\to \infty$. 
For nonnegative demi-supermartingales, the above limit always exists, as guaranteed by  \citet[Corollary 8]{NW82}.

\begin{lemma}[A Doob-like inequality for nonnegative demi-supermartingales]\label{lemma:demi-doob}
    Let $(M_k)_{k \in \Tcal}$ be a nonnegative demi-supermartingale. Fix $x > 0$ and define
    the first passage time \begin{equation}\label{eq:first-passage}
      \tau = \inf \{ k \in \Tcal : M_k \geq x \}  
    \end{equation}
    with the convention that $\tau = \infty$ if the set is empty. Then, $\EE[M_\tau] \leq \EE[M_0]$.
\end{lemma}
\begin{proof}
For $k \in \Tcal$ 
write $M_{\tau \land k}$ as 
\[
M_{\tau \land k}
=
M_0+\sum_{j=1}^{k}(M_{j}-M_{j-1})\mathbf 1\{\tau>j-1\}.
\]
Note that $\1 \{\tau>k\} =
\mathbf 1\left\{\max_{0\le j\le k}M_j<x\right\}$
and that the map
\(
(y_0,\dots,y_k)
\mapsto
\mathbf 1\left\{\max_{0\le j\le k}y_j<x\right\}
\)
is decreasing and nonnegative. Using the fact that \((M_k)_{k \in \Tcal}\) is a
demi-supermartingale, it follows that $\E\left[(M_{k+1}-M_k)\1 \{ \tau > k \}\right]\leq 0$.
Taking expectations, we obtain  
\(
 \EE[M_{\tau \land k}] \leq \EE[M_0].
\)
Appealing to nonnegativity of $(M_k)_{k \in \Tcal}$ and Fatou's lemma, we conclude $\EE[M_\tau] \leq \EE[M_0]$.
\end{proof}
Notice that if $(M_k)_{k \in \Tcal}$ is a supermartingale with respect to a filtration $\Fcal$ and $\tau$ is an $\Fcal$-stopping time, then \cref{lemma:demi-doob} follows from Doob's optional stopping theorem. \cref{lemma:demi-doob} is essentially a result of \citet{hadjikyriakou2025doob} but for demi-supermartingales. We now use \cref{lemma:demi-doob} in the proof of Ville's inequality for nonnegative demi-supermartingales.

\begin{lemma}[Ville's inequality for nonnegative demi-supermartingales]
\label{lem:ville}
Let $(M_k)_{k \in \Tcal}$ be a nonnegative demi-supermartingale. Then
\begin{equation}
\label{eq:new-ville}
\P\left(\sup_{k \in \Tcal} M_k \ge x\right)  \le \frac{\E[M_0]}{x}\quad\text{for all $x > 0$.}
\end{equation}
\end{lemma}

\begin{proof}
Fix $x > 0$ and consider the first passage $\tau = \inf \{ k \in \Tcal : M_k \geq x\}$ as well as the indicator $\1 \{ \tau > k \} = \1 \{ \max _{0 \leq j \leq k} M_j < x \}$ as in the proof of \cref{lemma:demi-doob}. On the event $\{ \tau < \infty \}$, it holds that $M_\tau \geq x$. On the event $\{\tau = \infty ,~ \sup_{k\in \mathcal T} M_k=x\}$, we have $M_\tau=M_\infty=x $. Taking expectations and applying \cref{lemma:demi-doob}, we have
\begin{equation}
 x \P\left(\sup_{k \in \Tcal} M_k \ge x\right) =   x \P \left( \tau < \infty \right )
    +x\P\left(  \tau = \infty ,~ \sup_{k\in \mathcal T} M_k=x \right) 
    \leq \EE [M_{\tau}] \leq \EE [M_0],
\end{equation}
from which \eqref{eq:new-ville} follows.  
\end{proof}
Note that in discrete time, the usual form of Ville's inequality is stated for nonnegative supermartingales and yet, every supermartingale is a demi-supermartingale. Therefore, \cref{lem:ville} is a strict generalization of Ville's inequality in discrete time, noting that demi-supermartingales are not necessarily supermartingales. A simple example is given below.

\begin{example}
    Let $E_1$ and $E_2$ be independent random variables satisfying
 $
    \mathbb{P}(E_i=0)=\mathbb{P}(E_i=2)=1/2$  for $i\in\{1,2\}.
 $  Let 
 $ 
    A_0=1$, $
    A_1=(E_1+E_2)/2,$ and 
   $ A_2=E_1E_2.
 $ 
 \cref{lem:weighted-sympol}  verifies that $(A_k)_{k=0}^2$ is a nonnegative
demi-martingale by taking $n=2$ and $a_1=a_2=1$.
It is not  a supermartingale with respect to its
natural filtration, because  
 $
    \mathbb{E}[A_2 \mid A_1]
    =4
    >2
    =A_1$ on 
 $ \{A_1=2\}.
 $
Thus $(A_k)_{k=0}^2$ is a nonnegative demi-martingale
that is not a supermartingale.
\end{example}

Our main motivation to study the aforementioned inequalities for demi-supermartingales stems from the fact that certain processes involved in the construction of $\sup_{\lambda \in [0,1]} M_n(\lambda)$ (see \eqref{eq:M}) turn out to be demi-supermartingales but not supermartingales, and this holds under a \emph{co-valid} dependence structure. Let us now introduce this structure and derive a maximal inequality for $\sup_{\lambda \in [0, 1]} M_n(\lambda)$ now that we have access to \cref{lem:ville}.

\section{Co-valid e-variables and the SymPol inequality}\label{section:co-validity}
We now define co-valid e-variables alongside the standard notions of
e-variables and sequential  e-variables.
 Let $[n]= \{1,\dots,n\}$ for a positive integer $n$.  
\begin{definition} 
\label{def:1}
Fix a collection of probability distributions $\Pcal$, which represents the null hypothesis. A random variable $E$ is   an \emph{e-variable} (for $\Pcal$) if it is $[0,\infty]$-valued and $ \EE_\P [E] \leq 1$ for $\p\in \Pcal$. 
The e-variables  $E_1, \dots, E_n$ 
are \emph{sequential} if 
$$\E_\p[E_i\mid E_1,\dots,E_{i-1}]\le 1\quad \text{for all } i\in [n],~\p\in \Pcal.$$
The e-variables $E_1,\dots,E_n$
are   \emph{co-valid}  if 
$$\E_\p[E_i\mid  E_1,\dots,E_{i-1}, E_{i+1},\dots,E_n]\le 1 \quad\text{for all } i\in [n],~\p\in \Pcal.$$
\end{definition}

In the literature, e-variables and their realizations are also called \emph{e-values}, and we do not distinguish these terms in this paper.
 
For the results in this section, it suffices to consider a single probability distribution $\P$ so we write $\E$ in place of $\EE_\P$; see \cite{VW21}. Strictly speaking, sequential e-variables can be defined on a filtration different from the natural filtration of $(E_1,\dots,E_n)$; see for example \citet[Definition 7.19]{RW25}.
The simplified version of sequential e-variables in \cref{def:1} was formulated in \citet[Section 4]{VW21}. 
It is easy to see that the following chain of inclusions holds:
\begin{equation}\label{eq:set-inclusions}
\mbox{\{independent e-variables\}}
\subseteq
\mbox{\{co-valid e-variables\}}
\subseteq
\mbox{\{sequential  e-variables\}}.
\end{equation} 
To elaborate on the term ``co-valid" e-variables, imagine  that $n$ labs run experiments to test a hypothesis, and each lab  generates an e-variable, which is valid regardless of the results of the other labs.
By contrast, sequential (``sequentially valid") e-variables would arise in a situation where each lab runs the experiment one by one, and every lab can design their experiment to generate an e-variable based on  the results from previous labs. 
For a concrete example, 
co-valid e-variables may be conditionally independent and conditionally valid on a common factor $Z$, as illustrated by the following proposition.

\begin{proposition} 
\label{prop:1}
Suppose that
    $E_1,\dots,E_n$  are  conditionally independent 
    and conditionally valid
    e-variables on a common variable $Z$.  Then $E_1,\dots,E_n$ are co-valid e-variables.
\end{proposition}
\begin{proof} 
For $i\in [n]$,  let $\mathbf E_{-i}=(E_j)_{j\ne i}$. 
By conditional independence,
$
\E[E_i \mid Z, \mathbf E_{-i}] = \E[E_i\mid Z],
$
and  $\E[E_i\mid Z] \le 1$ because $E_i$ is an e-variable conditionally on $Z$.  
Therefore,  
$$
\E[E_i\mid \mathbf E_{-i}]
= \E \left[\E[E_i\mid Z,\mathbf E_{-i}] \mid \mathbf E_{-i} \right]
= \E \left[\E[E_i\mid Z] \mid \mathbf E_{-i} \right] \le 1. 
$$  
Thus $E_1,\dots,E_n$ are co-valid e-variables.
\end{proof} 

With \cref{def:1} in mind, we are ready to state the main result of this section, the ``SymPol'' (\textbf{sym}metric \textbf{pol}ynomials) inequality for co-valid e-variables. 
For a vector $\bx \in \R^n$, let $k\in [n]$ and define $S_k(\mathbf x)$ as the 
elementary symmetric polynomial of degree $k$. That is,
$S_k(\bx)$ is the sum of all terms $\prod_{i\in I}x_i$
for subsets $I\subseteq [n]$ with size $k$.
We let $A_k(\mathbf x)$ be the average of the terms $\prod_{i\in I}x_i$ in $S_k(\bx)$:
\begin{equation}
\label{eq:sympol}
A_k(\mathbf{x)}=\frac{S_k(\mathbf x)}{\binom{n}{k}}
=\frac{1}{\binom{n}{k}}\sum_{I\subseteq [n],~|I|=k}\left(\prod_{i\in I}x_i \right),
\end{equation} 
with the convention $A_0=1$. The following theorem establishes that when considering an $n$-vector of co-valid e-variables $\bE$, the process $(A_k(\bE))_{k=0}^n$ forms a demi-supermartingale, and that it is an upper bound on the supremum of $\prod_{i=1}^n (1-\lambda + \lambda E_i )$ over $\lambda \in [0,1]$. 
The proof relies on a more general result (\cref{lem:weighted-sympol}), presented  in \cref{sec:heterogeneous-means}, that establishes the demi-supermartingality for a larger class of stochastic processes including $(A_k(\bE))_{k=0}^n$.
 
\begin{theorem} 
[SymPol inequality]
\label{theorem:maximal-inequality-co-valid}
Let $
\mathbf E=(E_1, \dots, E_n )$ be a vector of co-valid e-variables. 
Then, $(A_k(\bE))_{k=0}^n$ forms a nonnegative demi-supermartingale and hence
\begin{equation}
\label{eq:th1-1}
\p\left(\max_{0\le k\le n} A_k(\mathbf E) \ge x\right) \le \frac{1}{x}, \quad \mbox{ for all } x > 0.
\end{equation}
Furthermore, it holds that
\begin{equation}
\label{eq:th1-2} 
\sup_{\lambda \in [0,1]} \prod_{i=1}^n (1-\lambda + \lambda E_i )\le \max_{0\le k\le n} A_k(\mathbf E).
\end{equation}
\end{theorem}

\begin{proof}
    The proof of \eqref{eq:th1-1} follows by combining two results: the fact that $(A_k)_{k=0}^n$ is a nonnegative demi-supermartingale (taking $a_1=\cdots=a_n=1$ in  \cref{lem:weighted-sympol}) and   Ville's inequality for nonnegative demi-supermartingales (\cref{lem:ville}).
    The inequality \eqref{eq:th1-2} follows from standard algebra:  For every $\lambda\in[0,1]$, we have
\[
\prod_{i=1}^n\bigl(\lambda E_i+(1-\lambda)\bigr)
=\sum_{k=0}^n \binom{n}{k}\lambda^k(1-\lambda)^{n-k}A_k.
\]
Since the coefficients $\binom{n}{k}\lambda^k(1-\lambda)^{n-k}$ are nonnegative and sum to $1$, we obtain  \eqref{eq:th1-2}.
\end{proof}
Note that taking \eqref{eq:th1-1}--\eqref{eq:th1-2} together, we have that under co-validity (which includes the independent case),
\begin{equation}
\label{eq:corollary-conjecture}
\P\left(\sup_{\lambda \in [0,1]} \prod_{i=1}^n (1 - \lambda + \lambda E_i) \ge x\right) \le \frac{1}{x}, \quad \mbox{ for all } x > 0,
\end{equation}
providing a confirmation of \cref{conjecture:gwz} both with a sharper test statistic and under weaker assumptions than those conjectured.

Note that \eqref{eq:corollary-conjecture} cannot be shown by applying Ville's inequality to the process $\lambda \mapsto \prod_{i=1}^n (1 - \lambda + \lambda E_i)  $. 
If one equips it with its natural filtration, then after observing the process on any nontrivial interval the
entire polynomial, and hence its future, is measurable; supermartingality would therefore impose pathwise
nonincrease from every positive lambda. This generally fails.

\begin{remark} 
    As is apparent in the proof of \cref{theorem:maximal-inequality-co-valid}, if all of the e-variables $E_1, \dots, E_n$ are \emph{exact}, meaning that
    \begin{equation}
        \EE[E_i] = 1 \quad\text{for each $i \in [n]$},
    \end{equation}
    then $(A_k(\bE))_{k=0}^n$ forms a nonnegative demi-martingale. While this fact does not affect the form of the downstream tests and p-values (as is discussed shortly), it may lead to sharper inference in practice. Indeed, \cref{section:bdd-means} revisits the bounded mean estimation problem where confidence intervals are formed by inverting tests based on demi-martingales and exact e-variables.
\end{remark}

Given the set inclusions in \eqref{eq:set-inclusions}, it is natural to wonder whether it can be shown that the results of \cref{theorem:maximal-inequality-co-valid} fail to hold under sequential dependence of the e-variables. The following example demonstrates that this is in fact the case. 
\begin{example} 
\label{ex:counter1}
Take $n=2$. Let 
$E_1$ and $E_2$ be such that $\p(E_1=2)=\p(E_1=0)=1/2$,
and  
$\p(E_2=1\mid E_1=2)=1$,
and $\p(E_2=8 \mid E_1=0)=1-\p(E_2=0 \mid E_1=0)=1/8$. 
It is clear that $\E[E_1]=1$ and $\E[E_2\mid E_1]=1$, and hence they are sequential e-variables.
It is straightforward to compute $$\p\left(\sup_{\lambda \in [0,1]} \prod_{i=1}^2 (\lambda E_i + (1-\lambda) ) \ge 2\right)
=\p\big( (E_1,E_2) \in \{(2,1),(0,8)\}\big)
=\frac{9}{16}>\frac{1}{2},$$
which rules out \eqref{eq:corollary-conjecture}.
\end{example}

\cref{theorem:maximal-inequality-co-valid} immediately yields a new method for deriving hypothesis tests and p-values using e-variables.
In particular, let $\mathbf E=(E_1,\dots,E_n)$ be a vector of co-valid e-variables and define
\begin{equation}
    P_n = \inf_{\lambda \in [0,1]} \prod_{i=1}^n \frac{1}{(1-\lambda + \lambda E_i)} \quad \text{and} \quad \widebar P_n = \min_{0 \leq k \leq n} \frac{1}{A_k(\bE)}.
\end{equation}
By \cref{theorem:maximal-inequality-co-valid}, $P_n$ and $\widebar P_n$ are both p-values, i.e.,
\begin{equation}
    \forall \alpha \in (0, 1),\quad \P \left ( P_n \leq \alpha \right ) \leq \P \left ( \widebar P_n \leq \alpha \right ) \leq \alpha.
\end{equation}
Consequently, for any desired significance level, both $\1 \{ P_n \leq \alpha \}$ and $\1 \{ \widebar P_n \leq \alpha \}$ form level-$\alpha$ hypothesis tests for the null that the components of $\bE$ are co-valid e-variables. The computational aspects of these different tests are discussed in \cref{sec:comp}. 

An important class of hypothesis tests consists of those that can be inverted to form confidence sets for the parameters being tested. In the following section, we explore an instantiation of such tests for the purposes of deriving confidence intervals for means of bounded random variables.

\section{Confidence sets for means of bounded random variables}\label{section:bdd-means}

\subsection{Problem setup}
The construction of confidence intervals for means of bounded random variables is a foundational problem in statistical inference. It is a key component in several methodological problems including the derivation of generalization bounds \citep{maurer2009empirical}, prediction-powered inference \citep{angelopoulos2023prediction}, risk-limiting election audits \citep{stark2008conservative,stark2020sets,stark2023alpha,waudby2021rilacs}, risk-controlling conformal prediction sets \citep{bates2021distribution}, sensitivity analysis of individual treatment effects \citep{jin2023sensitivity}, among other applications.

Making matters concrete, let $\mathbf X=(X_1, \dots, X_n)$ be a random vector and let $\Pcal$ be the collection of probability measures for which $X_1, \dots, X_n$ are independent,\footnote{The independence assumption can be relaxed to a rescaled version of co-validity as in \cref{theorem:maximal-inequality-co-valid}.} are supported on $[0,1]$, and have the common mean  $\mu = \mu(\P) = \EE_\P[X_i]$ for all $i \in [n]$.
Denote their sample mean  by $\overline X_n$. 
For a fixed $\alpha \in (0, 1)$, we are interested in deriving a confidence interval for $\mu$, which is a random set $C_n(\alpha)$, formed from $X_1, \dots, X_n$, for which 
\begin{equation}
    \inf_{\P \in \Pcal} ~ \P \left ( \mu(\P) \in C_n(\alpha) \right ) \geq 1-\alpha.
\end{equation}
With access to \cref{theorem:maximal-inequality-co-valid}, we now study the lower and upper confidence intervals formed when they are constructed using
\begin{equation}
    W_n^\brackell(m) = \sup_{\gamma \in [0,1/m]} \prod_{i=1}^n (1 + \gamma (X_i - m))
\end{equation}
and 
\begin{equation}
    W_n^\bracku(m) = \sup_{\gamma \in [-1/(1-m),0]} \prod_{i=1}^n (1 + \gamma (X_i - m)),
\end{equation}
respectively. Throughout we set $x/0=\infty$ for $x>0$, and the interval $[0,1/m]$ for $m=0$ is interpreted as $[0,\infty)$.
We next state the conclusion formally in the following corollary.
\begin{corollary}\label{corollary:supgamma-ci}
    Fix $\alpha \in (0, 1)$ and define the lower and upper intervals $L_n(\alpha)$ and $U_n(\alpha)$ with their endpoints:
    \begin{equation}
        L_n = \inf \left \{ m \in [0, 1]  : W_n^\brackell(m) < \frac{2}{\alpha} \right \}\quad \text{and}\quad U_n = \sup \left \{ m \in [0, 1]  : W_n^\bracku(m) < \frac{2}{\alpha} \right \}.
    \end{equation}
    Let $C_n= (L_n , U_n )$ 
if $\overline{X}_n\not\in\{0,1\}$, $C_n = [ L_n ,U_n )$  if $\overline{X}_n=0$,
 and  $C_n =(L_n , U_n ]$ if $\overline{X}_n=1$.  
 Then it holds that $C_n$ is a valid confidence interval for $\mu\in[0,1]$ that contains the sample mean $\overline X_n$. 
\end{corollary}
\cref{corollary:supgamma-ci} is immediate once we identify $\lambda \in [0, 1]$ with $m \gamma$ in $L_n$ and with $-(1-m) \gamma$ in $U_n$ and apply \cref{theorem:maximal-inequality-co-valid}.
\begin{remark} \label{remark:computational-feasibility}
It is straightforward to compute each of the suprema in the definitions of $L_n$ and $U_n$ for any $m \in[0,1]$ via standard convex optimization (or root-finding) routines. It is less obvious, however, whether the endpoints $L_n$ and $U_n$ themselves are computable. Note that $W_n^\brackell(m)$ can be written as
  \begin{equation}\label{eq:fixed lambda product}
    W_n^\brackell(m) = \sup_{\lambda \in [0, 1]}\prod_{i=1}^n \left ( 1 - \lambda + \lambda \frac{X_i}{m} \right ),
  \end{equation}
   and the right-hand side is clearly a supremum over functions that are convex in $m\in[0,1]$. It follows that $(L_n, 1]$ is the intersection of infinitely many intervals, one for each $\lambda \in [0, 1]$. Consequently, $W_n^\brackell(m)$ (and through a similar argument, $W_n^\bracku(m)$) are quasi-convex in $m \in [0, 1]$ so $L_n$ and $U_n$ can be computed via standard root-finding algorithms.
\end{remark}

In the same way that $C_n$ is the result of appropriately inverting the tests given by $\1 \{ W_n^\brackell(m) \geq 2/\alpha \}$  and $\1 \{ W_n^\bracku(m) \geq 2/\alpha \}$, it is theoretically possible to do so with the corresponding statistics given by $\max_{0\leq k \leq n} A_k(\bE)$ in \cref{theorem:maximal-inequality-co-valid}. However, it is not obvious if or how one could find the endpoints of such a set (whether theoretically or practically). The following proposition shows that the resulting confidence interval has a closed-form expression and does not require any root-finding or optimization to solve for the endpoints. We refer to them as SymPol confidence intervals.
\begin{proposition}[SymPol confidence intervals]\label{proposition:sympol-ci}
Consider the function $A_k $ as in \eqref{eq:sympol} and define
\begin{align}
 L_n^{\SP} &=
\max_{1\le k\le n}
\left(
\frac{\alpha}{2}\,A_k(X_1,\dots,X_n)
\right)^{1/k}, \text{ and}\\
U_n^{\SP} &= 1-
\max_{1\le k\le n}
\left(
\frac{\alpha}{2}\,A_k(1-X_1,\dots,1-X_n)
\right)^{1/k}.
\end{align}
Let $C_n^{\rm SP}= (L_n^{\SP}, U_n^{\SP})$ 
if $\overline{X}_n\not\in\{0,1\}$, $C_n^{\rm SP}= [ L_n^{\SP},U_n^{\SP})$  if $\overline{X}_n=0$,
 and  $C_n^{\rm SP}=(L_n^{\SP}, U_n^{\SP}]$ if $\overline{X}_n=1$.  
Then $C_n^{\rm SP}$ is a $(1-\alpha)$-confidence interval for the mean:
\begin{equation}
\inf_\Pin \P \left ( \mu(\P) \in C_n^{\rm SP} \right ) \geq 1-\alpha.
\end{equation}
Furthermore, $L_n \leq L_n^{\SP}\le\widebar{X}_n\le U_n^{\SP} \leq U_n$. 
\end{proposition}
The proof of \cref{proposition:sympol-ci} is given in \cref{proof:proposition:sympol-ci}.
The statistics given by $W_n^\brackell(m)$ and $W_n^\bracku(m)$ have appeared in different contexts in the literature on best-arm identification and on sequential estimation. Let us now make these connections explicit. 

\subsection{A regret-free KL-inf confidence interval}
A fundamental quantity that arises in the multi-armed bandit literature is a certain infimum over Kullback--Leibler (KL) divergences that is colloquially referred to as ``$\KLinf$''. See \citet{burnetas1996optimal}, \citet{honda2010asymptotically,honda2015non}, and \citet{agrawal2021regret,agrawal2021optimal}. Let us recall and discuss this quantity formally here.
\begin{definition}[$\KLinf$]
  Let $\mathcal S$ be a collection of probability distributions (to be thought of as a composite null hypothesis) and $\Q$ a single probability distribution (to be thought of as a point alternative). The quantity $\KLinf(\Q \| \mathcal S)$ is defined as
  \begin{equation}
    \KLinf(\Q \| \mathcal S) = \inf_{\P \in \mathcal S} \KL(\Q \| \P),
  \end{equation}
  where $\KL(\Q \| \P)$ is the KL divergence between $\Q$ and $\P$, defined as \begin{equation}
\label{eq:kl-def1}
\mathrm{KL}(\Q\Vert \p) =  \E_\Q\left[\log \frac{\d \Q}{\d\p} \right]  \mbox{ if $\Q\ll\p$, and   $\infty$ otherwise. }
\end{equation}

\end{definition}
\citet{honda2010asymptotically} derived a duality result that relates $W_n^\brackell(m)$ in~\eqref{eq:fixed lambda product} to the exponential of the $\KLinf$ between the empirical measure of the data $X_1, \dots, X_n$ and the null distributions with mean  at most $m$, recalled below. 

\begin{fact}[Duality between wealth and KL-inf]\label{fact:honda-takemura}
For $\mathbf{x} = (x_1, \dots, x_n) \in [0, 1]^n$, denote their empirical measure by $\widehat \Q_n(\mathbf{x}) := \frac{1}{n}\sum_{i=1}^n \delta_{x_i}$.
  For a fixed $m \in [0, 1]$, let $$\Pcal_m^\leq := \left\{ \Q   \in \Pcal([0,1]) : \int x \d \Q  (x) \leq m  \right\},$$
  where $\Pcal([0,1]) $ 
  is the set of distributions  supported on $[0,1]$.
  Then,   $\log     W_n^\brackell(m)    = n\KLinf(\widehat \Q_n(\mathbf{x}) \| \Pcal_m^\leq)$.
\end{fact}

Combining the above fact with our main theorem, we obtain the following regret-free concentration for KL-inf.

\begin{corollary}[Regret-free concentration for KL-inf]\label{corollary:klinf-concentration}
Using the same notation as above, for independent $X_1,\dots,X_n$ with mean no larger than $m$,
  \begin{equation}\label{eq:klinf-inequality}
\P \left ( n\KLinf (\widehat \Q_n (\mathbf X)\| \Pcal_m^\leq) \geq \log \frac{1}{\alpha} \right ) \leq \alpha.
  \end{equation}
In particular, $\inf  \{ m \in [0,1] : n\KLinf (\widehat \Q_n (\mathbf X) \| \Pcal_m^\leq ) < \log (1 /\alpha) \}$
  is a lower $(1-\alpha)$-confidence set for the mean. 
\end{corollary}
Upper confidence sets can be derived analogously.
When written in the above form, \cref{corollary:klinf-concentration} may come as a surprise to readers who are familiar with the KL-inf. In particular, it is commonly seen in the bandit and sequential testing literature that KL-inf statistics concentrate with a similar threshold to $\log(1/\alpha)$ but inflated by a regret bound; e.g., \citep{agrawal2021regret}. For example, it can be deduced from \citet{orabona2023tight} (see also \citet[Corollary 2.2]{waudby2025universal}) combined with the duality results of \citet{honda2010asymptotically} that
\begin{equation}\label{eq:regret-based-conc}
  \P \left ( n\KLinf(\widehat \Q_n (\mathbf X) \| \Pcal_m^\leq) \geq \log  \frac{1}{\alpha}   + \frac{\log (n + 1)}{2} + \log (2) \right ) \leq \alpha.
\end{equation}
The essential reason   why \eqref{eq:regret-based-conc} holds is that $\exp \{ n\KLinf(\widehat \Q_n (\mathbf X) \| \Pcal_m^\leq ) - \log(n+1) / 2 - \log(2) \}$ being upper-bounded by an e-variable, to which Markov's inequality can be applied. This approach is taken explicitly for the derivation of confidence intervals in \citet{orabona2023tight} and for sequential hypothesis tests in \citet{waudby2025universal}. It should be noted that the aforementioned works that rely on inequalities of the kind in \eqref{eq:regret-based-conc} enjoy \emph{anytime validity}, meaning that inequalities hold when the sample size $n$ is replaced by a data-dependent stopping time. As a consequence, the widths of the resulting confidence intervals cannot scale at the typical $1/\sqrt{n}$ rate (see, for example, discussions in \citep{howard2021time,WR24}). It is therefore natural to ask: do the confidence intervals resulting from the inversion of the sharper inequality \eqref{eq:klinf-inequality} scale at a rate of $1/\sqrt{n}$, and if so, with what constants? The next section answers some of these questions.

\subsection{Asymptotic widths of the KL-inf confidence interval}

For the purposes of studying asymptotic widths of confidence intervals, we assume that the random variables $X_1, \dots, X_n$ are independent and identically distributed (iid) from some distribution $\P$, and the asymptotic behavior of these intervals will be characterized in terms of properties of $\P$.
Define $\mu := \E_\P[X_1] \in [0, 1]$ and $\sigma := \sqrt{\Var_\P[X_1]} \in [0, 1/2]$. 
Define
\[
    \overline W_n(m)
    :=
    \sup_{\gamma\in[-1/(1-m),\,1/m]}
    \prod_{i=1}^n \bigl(1+\gamma(X_i-m)\bigr)
    =\max\left\{W_n^\brackell(m),W_n^\bracku(m)\right\},
\]
with the convention that $1/0=\infty$ when $m$ reaches the endpoints $\{0, 1\}$. For $\alpha\in(0,1)$, define
\[
    C_n
    =
    \left\{m\in[0,1]:\overline W_n(m)<\frac{2}{\alpha}\right\}
    =\left\{W_n^\brackell(m)<\frac{2}{\alpha}\right\}\cap\left\{W_n^\bracku(m)<\frac{2}{\alpha}\right\}.
\]
Let
\(
    \width_n := \sup C_n-\inf C_n
\)
denote the length of the confidence interval $C_n$. With this notation in mind, we have the following theorem.

\begin{theorem}\label{thm:width}
Suppose that $\sigma^2>0$. Then the confidence interval $C_n$ has a length scaling asymptotically as
\[
    \sqrt n\,\width_n
    \longrightarrow
    2\sigma\sqrt{2\log(2/\alpha)}\quad\text{almost surely.}
\]
If $\sigma^2=0$,
then
$\width_n=1- (\alpha/2)^{1/n}$
and hence
$n\width_n \to \log(2/\alpha)$.
\end{theorem}

The proof is given in \cref{proof:thm:width}. The key technical tool for the above theorem is a uniform Taylor approximation, stated below and proved in \cref{proof:lem:klinf-asymp}. Since the SymPol sets are always subsets of those based on the KL-inf statistic, their widths are also upper bounded by the same expressions as above. \cref{app:sympol-size} shows that SymPol has the same limiting width as derived above. 

\begin{lemma}\label{lem:klinf-asymp}
For $m\in[0,1]$, define
\[
    \ell_n(m):=\log \overline W_n(m)
    =
    \sup_{\gamma\in[-1/(1-m),\,1/m]}
    \sum_{i=1}^n
    \log\bigl(1+\gamma(X_i-m)\bigr).
\]
Assume $\sigma^2>0$. Near
$\overline X_n$, the profile $\ell_n(m)$ is asymptotically quadratic:
For every fixed $K<\infty$, almost surely we have
\[
\sup_{|t|\le K}
\left|
\ell_n\left(\overline X_n+\frac{t}{\sqrt n}\right)
-\frac{t^2}{2\sigma^2}
\right|
\longrightarrow0.
\]
\end{lemma}

\subsection{A   general KL-inf  duality theorem}
The concentration inequality for the KL-inf statistic in \cref{corollary:klinf-concentration} relied on a combination of \cref{theorem:maximal-inequality-co-valid} with a duality result of \citet{honda2010asymptotically}. However, the latter is specific to the problem of testing means of bounded random variables while the former holds for generic nonnegative random variables with bounded means. Here, we prove a KL-inf duality theorem for the generic nonnegative setting, arriving at a generalization of \cref{corollary:klinf-concentration}.
\begin{theorem}
\label{th:KL-inf}
For $\mathbf{x} = (x_1, \dots, x_n) \in [0, \infty)^n$, denote their empirical measure by $\widehat \Q_n(\mathbf{x}) := \frac{1}{n}\sum_{i=1}^n \delta_{x_i}$. Define the set of nonnegative distributions with mean at most one
\[
  \mathcal C := \left\{\Q \in \Pcal([0,\infty)): \int x \d\Q (x) \le 1 \right\},
\]
where $\Pcal([0,\infty))$ is the set of probability measures supported on $[0,\infty)$.
 It holds that
\[
  \sup_{\lambda \in[0,1]} \frac{1}{n}\sum_{i=1}^n \log(1-\lambda+\lambda x_i)
  =
  \inf_{\Q \in\mathcal C}\KL(\widehat \Q_n(\mathbf x) \|\Q) \equiv \KLinf(\widehat \Q_n(\mathbf x) \| \mathcal C).
\]
\end{theorem}
\cref{th:KL-inf} can be viewed as a generalization of \cref{fact:honda-takemura}. The proof of \cref{th:KL-inf} can be found in \cref{proof:th:KL-inf}.
By \cref{theorem:maximal-inequality-co-valid}, we observe that  for iid nonnegative random variables $Y_1, \dots, Y_n$ that are e-variables under $\P$, it holds that
\begin{equation}
    \P \left ( n\KLinf(\widehat \Q_n(Y_1, \dots, Y_n) \| \mathcal C) \geq \log \frac{1}{\alpha} \right ) \leq \alpha.
\end{equation}

\section{Extension to heterogeneous means}
\label{sec:heterogeneous-means}

In this section, we show the stronger result that
the SymPol inequality remains valid without the common-mean assumption.
In statistical terms, our results apply to compound e-values \citep{ignatiadis2024asymptotic}, which is a relaxation of e-values. Putting the terminology of \cite{ignatiadis2024asymptotic} into our context,   nonnegative random variables $E_1,\dots,E_n$ are \emph{compound e-values (compound e-variables)}
if 
$$
\frac{1}{n}\sum_{i=1}^n \E[E_i]\le 1,
$$
and, by naturally extending \cref{def:1}, we say they are \emph{co-valid} if 
$$
\E[E_i\mid \mathbf E_{-i}]\le 
\mu_i\in \R,~ \mbox{where $\mathbf E_{-i}=(E_j)_{j\ne i}$ for $i\in[n]$}, \quad \mbox{and}\quad 
\frac{1}{n}\sum_{i=1}^n \mu_i \le 1.
$$
\cite{ignatiadis2024asymptotic} considered the more general context of multiple testing with different hypotheses, and the above formulation corresponds to the simpler case of testing the global null.

We first record a weighted version of the SymPol inequality under co-validity, which is the key lemma that proves \cref{theorem:maximal-inequality-co-valid}.

\begin{lemma}[SymPol demi-(super)martingality]\label{lem:weighted-sympol} 
Let $
\mathbf E=(E_1, \dots, E_n )$ be a vector of co-valid e-variables and $a_1,\ldots,a_n>0$. Define
\[
B_k
=
\frac{S_k(a_1E_1,\dots,a_nE_n)}
     {S_k(a_1,\dots,a_n)},
 \quad k=1,\dots,n,\quad\mbox{with $B_0=1$.}
\]
Then $(B_k)_{k=0}^n$ is a demi-supermartingale.
If, further, each component of $\mathbf E$ has mean $1$, then $(B_k)_{k=0}^n$ is a demi-martingale. 
\end{lemma}

\cref{theorem:maximal-inequality-co-valid} immediately follows by applying Ville's inequality for nonnegative demi-supermartingales in   \cref{lem:ville} to $(B_k)_{k=0}^n$ in \cref{lem:weighted-sympol}: 
\[
\mathbb P\left( 
\max_{0\le k\le n}B_k\ge x
\right)
\le \frac1x,
\quad x>0.
\]

\begin{proof}
[Proof of Lemma~\ref{lem:weighted-sympol}]
For $S\subseteq[n]$, write
\[
a_S=\prod_{i\in S}a_i,
\quad
E_S=\prod_{i\in S}E_i,
\]
with the convention $a_\varnothing=E_\varnothing=1$.
We first note that all quantities below are integrable. Indeed, for every
nonempty $S\subseteq[n]$ and every $i\in S$,
\[
\E[E_S]
=
\E\left[
E_{S\setminus\{i\}}
\E[E_i\mid \mathbf E_{-i}]
\right]
\le
\E[E_{S\setminus\{i\}}].
\]
Iterating gives
\(
\E[E_S]\le1.
\)
In particular, each $B_k$ is integrable, and for $i\notin S$,
\[
\E[(E_i+1)E_S]
=
\E[E_{S\cup\{i\}}]+\E[E_S]
\le2.
\]

For $|S|=k$, define
\[
q_k(S)=\frac{a_S}{S_k(a_1,\ldots,a_n)}.
\]
Thus
\[
B_k=\sum_{|S|=k}q_k(S)E_S.
\]
To construct compatible couplings of the laws $q_k$, let
$Z_1,\ldots,Z_n$ be independent Bernoulli random variables with odds $a_i$,
that is,
\[
\mathbb P(Z_i=1)=\frac{a_i}{1+a_i}.
\]
If $R=\{i:Z_i=1\}$, then the conditional law of $R$ given $|R|=k$ is
$q_k$. The product Bernoulli law is strong Rayleigh. Hence, by the
stochastic-covering property of adjacent rank conditionals
\cite[Proposition~2.3]{pemantle2014concentration}, for each $k<n$ there is a
coupling $(I_k,I_{k+1})$ such that
\begin{equation}
\label{eq:monotone-subset-coupling}
I_k\sim q_k,
\quad
I_{k+1}\sim q_{k+1},
\quad
I_k\subseteq I_{k+1}
\quad\text{almost surely}.
\end{equation}
Define
\[
\gamma_k(S,i)
=
\mathbb P\{I_k=S,\ I_{k+1}=S\cup\{i\}\},
\quad |S|=k,\quad i\notin S.
\]
These coefficients are nonnegative, and the two marginals of the coupling give
\begin{equation}
\label{eq:marginals}
\sum_{i\notin S}\gamma_k(S,i)=q_k(S),
\quad |S|=k,
\end{equation}
and
\begin{equation}
\label{eq:marginals-upper}
\sum_{i\in T}\gamma_k(T\setminus\{i\},i)=q_{k+1}(T),
\quad |T|=k+1.
\end{equation}
Consequently,
\[
B_{k+1}
=
\sum_{\substack{|S|=k\\i\notin S}}
\gamma_k(S,i)E_{S\cup\{i\}},
\quad B_k
=
\sum_{\substack{|S|=k\\i\notin S}}
\gamma_k(S,i)E_S.
\]
Subtracting the two identities yields
\begin{equation}
\label{eq:weighted-sympol-increment}
B_{k+1}-B_k
=
\sum_{\substack{|S|=k\\i\notin S}}
\gamma_k(S,i)
\left(E_{S\cup\{i\}}-E_S\right)
=
\sum_{\substack{|S|=k\\i\notin S}}
\gamma_k(S,i)E_S(E_i-1).
\end{equation}

Let
\(
G=g(B_0,\ldots,B_k),
\)
where $g$ is bounded, coordinatewise decreasing, and nonnegative. Since every
$B_j$ is coordinatewise nondecreasing in $(E_1,\ldots,E_n)$, the random
variable $G$ is nonincreasing in each $E_i$. The conditional Chebyshev's association inequality gives
\begin{equation}
\label{eq:conditional-drift-nonpositive}
\E[(E_i-1)G\mid \mathbf E_{-i}]
\le
\E[E_i-1\mid \mathbf E_{-i}]\E[G\mid \mathbf E_{-i}]
\le0.
\end{equation}
Because $i\notin S$, the variable $E_S$ is nonnegative and measurable with
respect to $\mathbf E_{-i}$. Multiplying
\eqref{eq:conditional-drift-nonpositive} by $E_S$ and taking expectations gives
\[
\E[E_S(E_i-1)G]\le0.
\]
Using \eqref{eq:weighted-sympol-increment} and
$\gamma_k(S,i)\ge0$, we conclude that
\[
\E\left[
(B_{k+1}-B_k)g(B_0,\ldots,B_k)
\right]
\le0.
\]
The boundedness restriction on $g$ can be removed by truncation whenever the
relevant expectation is defined. Thus $(B_k)_{k=0}^n$ is a nonnegative demi-supermartingale.

Suppose now that $\mathbb E[E_i]=1$ for every $i$. Since
$\E[E_i\mid \mathbf E_{-i}]\le1$ almost surely and $\mathbb E[\E[E_i\mid \mathbf E_{-i}]]=\mathbb E[E_i]=1$, we have
$\E[E_i\mid \mathbf E_{-i}]=1$ almost surely. The second term in
\eqref{eq:conditional-drift-nonpositive} then vanishes, so the same argument
holds for every bounded coordinatewise decreasing $g$, without requiring
$g\ge0$. Hence $(B_k)_{k=0}^n$ is a demi-martingale. 
\end{proof}

We now provide the  upgrade of \cref{theorem:maximal-inequality-co-valid} to nonnegative random variables with heterogeneous means, which is the strongest result of this paper.
We continue to use the shorthand symbols $\mathbf X=(X_1,\dots,X_n)$ and $\mathbf X_{-i}=(X_j)_{j \ne i}$.

\begin{theorem}[SymPol inequality for co-valid compound e-values]
\label{thm:average-mean-sympol}

Let $X_1,\ldots,X_n$ be nonnegative random variables such that 
 $
\E[X_i\mid \mathbf X_{-i}] \le \mu_i \mbox{ for all $i\in [n]$}
 $ 
 for some constants $\mu_1,\dots,\mu_n\ge 0 $,
 and 
$
\bar\mu =\sum_{i=1}^n\mu_i/n.
$
Throughout, we use the convention $0/0=1$.
Then 
 $$
\mathbb P\left(
\max_{0\le k\le n}
\frac{A_k(\mathbf X)}{\bar{\mu}^k}
\ge x
\right)
\le\frac1x,
\quad x>0. $$
\end{theorem}
 
If $X_1,\dots,X_n$ are independent, then $\mu_i$ in \cref{thm:average-mean-sympol} can be chosen as the mean of $X_i$. 
Note that $E_i=X_i/\E[X_i]$ 
is an e-variable for each $i$, with the convention $0/0=1$. If $E_1,\dots,E_n$ are co-valid e-variables, then the condition on $\mathbf X$ in \cref{thm:average-mean-sympol} holds.
In particular, if $\mu_1=\dots=\mu_n=1$, then \cref{thm:average-mean-sympol} recovers the main conclusion of \cref{theorem:maximal-inequality-co-valid}.
Moreover,
if $X_1,\dots,X_n$ are co-valid compound e-values, then we can take $\bar \mu=1$.
Therefore,  the SymPol inequality, 
 $$
\mathbb P\left(
\max_{0\le k\le n}
 {A_k(\mathbf E)}
\ge x
\right)
\le\frac1x,
\quad x>0, $$
holds for co-valid compound e-values $E_1,\dots,E_n$.

\begin{proof}
[Proof of Theorem~\ref{thm:average-mean-sympol}]
First observe that if $\mu_i=0$, then
\(
0
\le
\E[X_i]
=
\E\!\left[\E[X_i\mid\mathbf X_{-i}]\right]
\le
\mu_i
=
0.
\)
Since $X_i$ is nonnegative, it follows that
\begin{equation}
\label{eq:zero-bound-zero-variable}
X_i=0
\quad\text{almost surely whenever }\mu_i=0.
\end{equation}

If $\bar\mu=0$, then $\mu_i=0$ for every $i$, and hence
$\mathbf X=\mathbf 0$ almost surely by
\eqref{eq:zero-bound-zero-variable}. Under the convention $0/0=1$,
\[
\max_{0\le k\le n}
\frac{A_k(\mathbf X)}{\bar\mu^k}
=
1
\quad\text{almost surely}.
\]
Consequently, for every $x>0$,
\[
\mathbb P\left(
\max_{0\le k\le n}
\frac{A_k(\mathbf X)}{\bar\mu^k}
\ge x
\right)
=
\mathbf 1_{\{x\le1\}}
\le \frac1x.
\]
Thus it remains to consider the case $\bar\mu>0$.

For $i\in[n]$, define
\begin{equation}
\label{eq:normalized-compound-e-values}
E_i=\frac{X_i}{\mu_i},
\quad
a_i=\frac{\mu_i}{\bar\mu},
\end{equation}
where $E_i=1$ when $\mu_i=0$. By
\eqref{eq:zero-bound-zero-variable}, this convention is consistent,
and
\begin{equation}
\label{eq:scaled-coordinate-factorization}
a_iE_i=\frac{X_i}{\bar\mu}
\quad\text{almost surely for every }i\in[n].
\end{equation}

We next verify that $\mathbf E=(E_1,\ldots,E_n)$ is a vector of
co-valid e-variables. If $\mu_i=0$, then $E_i=1$ almost surely and
hence
\(
\E[E_i\mid\mathbf E_{-i}]=1.
\)
If $\mu_i>0$, then
\[
\E[E_i\mid\mathbf X_{-i}]
=
\frac{1}{\mu_i}
\E[X_i\mid\mathbf X_{-i}]
\le1
\quad\text{almost surely}.
\]
Moreover,
\(
\sigma(\mathbf E_{-i})
\subseteq
\sigma(\mathbf X_{-i}),
\)
because every $E_j$ is a deterministic function of $X_j$. Therefore,
by the tower property,
\begin{equation}
\E[E_i\mid\mathbf E_{-i}]
=
\E\!\left[
  \E[E_i\mid\mathbf X_{-i}]
  \,\middle|\,
  \mathbf E_{-i}
\right]
\le1
\quad\text{almost surely}.
\label{eq:normalized-covalidity-proof}
\end{equation}
Thus $\mathbf E$ is co-valid.

For $k=0,\ldots,n$, define
\begin{equation}
\label{eq:B-with-zero-weights}
B_k
=
\frac{S_k(a_1E_1,\ldots,a_nE_n)}
     {S_k(a_1,\ldots,a_n)},
\end{equation}
again using the convention $0/0=1$. Although
\cref{lem:weighted-sympol} is stated for strictly positive
weights, its maximal conclusion applies here as well. Indeed, omit
the coordinates for which $a_i=0$ and apply
\cref{lem:weighted-sympol} to the remaining positive weights.
Co-validity is preserved when passing to a subvector, by another
application of the tower property.

If $r$ denotes the number of positive weights, then
$S_k(a_1,\ldots,a_n)>0$ for $0\le k\le r$, and the quantities in
\eqref{eq:B-with-zero-weights} agree with those obtained after
omitting the zero-weight coordinates. For $k>r$, both the numerator
and denominator in \eqref{eq:B-with-zero-weights} vanish, so that
$B_k=1$ by convention. Since $B_0=1$, these additional values do not
change the maximum. Hence \cref{lem:weighted-sympol} and Ville's
inequality give
\begin{equation}
\label{eq:B-maximal-current-theorem}
\mathbb P\left(
\max_{0\le k\le n}B_k\ge x
\right)
\le\frac1x,
\quad x>0.
\end{equation}

By \eqref{eq:scaled-coordinate-factorization}, for every
$k=0,\ldots,n$,
\begin{align}
\frac{A_k(\mathbf X)}{\bar\mu^k}
&=
A_k\left(
\frac{X_1}{\bar\mu},
\ldots,
\frac{X_n}{\bar\mu}
\right)
=
\frac{
S_k(a_1E_1,\ldots,a_nE_n)
}{
\binom nk
}
\nonumber\\
&=
\frac{
S_k(a_1,\ldots,a_n)
}{
\binom nk
}
B_k
=
A_k(a_1,\ldots,a_n)B_k.
\label{eq:compound-factorization}
\end{align}
Applied to the nonnegative vector
$(a_1,\ldots,a_n)$, Maclaurin's inequality gives, for every $1\le k\le n$,
\begin{equation}
A_k(a_1,\ldots,a_n)^{1/k}
\le
A_1(a_1,\ldots,a_n)
=
\frac1n\sum_{i=1}^n\frac{\mu_i}{\bar\mu}
=
1.
\label{eq:maclaurin-current-theorem}
\end{equation}
Consequently,
\(
A_k(a_1,\ldots,a_n)\le1
\)
for
\(
k=1,\ldots,n.
\)
Together with \(A_0=1\), equation
\eqref{eq:compound-factorization} therefore implies
\[
\frac{A_k(\mathbf X)}{\bar\mu^k}
\le B_k,
\quad k=0,\ldots,n.
\]
It follows that
\[
\max_{0\le k\le n}
\frac{A_k(\mathbf X)}{\bar\mu^k}
\le
\max_{0\le k\le n}B_k
\quad\text{almost surely}.
\]
Combining this inequality with
\eqref{eq:B-maximal-current-theorem} yields
\[
\mathbb P\left(
\max_{0\le k\le n}
\frac{A_k(\mathbf X)}{\bar\mu^k}
\ge x
\right)
\le
\mathbb P\left(
\max_{0\le k\le n}B_k\ge x
\right)
\le\frac1x,
\]
which proves the result.
\end{proof}

As a consequence of \cref{thm:average-mean-sympol}, for $\theta>0$, the statistic
\begin{equation}
\label{eq:average-mean-p-value}
p_{\mathrm{SP,av}}(\theta;\mathbf x)
=
\frac{1}{
\displaystyle
\max_{0\le k\le n}A_k(\mathbf x/\theta)}
=
\frac{1}{
\displaystyle
\max_{0\le k\le n}\{A_k(\mathbf x)/\theta^k\}}
\end{equation}
is a valid p-value for testing
\[
H_{0,\le}(\theta):
\frac1n\sum_{i=1}^n\E[X_i]\le\theta
\]
 for nonnegative independent data (more generally, under the condition in \cref{thm:average-mean-sympol}). 

\paragraph{Confidence intervals for bounded variables.}
Suppose now that $X_i\in[0,1]$ and that the conditional mean bounds
in \cref{thm:average-mean-sympol} are sharp. More precisely,
let
\(
\mu_i:=\E[X_i],
\)
and assume that
\begin{equation}
\label{eq:conditional-mean-independence}
\E[X_i\mid\mathbf X_{-i}]
=
\mu_i
\qquad\text{almost surely for every }i\in[n].
\end{equation}
Write
\[
\bar\mu
=
\frac1n\sum_{i=1}^n\mu_i
=
\frac1n\sum_{i=1}^n\E[X_i].
\]

Choose $\alpha_L,\alpha_U\ge0$ with
$\alpha_L+\alpha_U\le\alpha<1$. For
$\mathbf x=(x_1,\ldots,x_n)\in[0,1]^n$, define
\[
L_{\mathrm{SP,av}}(\mathbf x)
=
\max_{1\le k\le n}
\left\{\alpha_L A_k(\mathbf x)\right\}^{1/k}
\]
and
\[
U_{\mathrm{SP,av}}(\mathbf x)
=
1-
\max_{1\le k\le n}
\left\{\alpha_U A_k(\mathbf 1-\mathbf x)\right\}^{1/k},
\]
where
\(
\mathbf 1-\mathbf x
=
(1-x_1,\ldots,1-x_n).
\)

Set $\mathbf Y=\mathbf 1-\mathbf X$. Since the coordinatewise
transformation $x\mapsto1-x$ is bijective,
\(
\sigma(\mathbf Y_{-i})
=
\sigma(\mathbf X_{-i}).
\)
Consequently, \eqref{eq:conditional-mean-independence} gives
\[
\E[Y_i\mid\mathbf Y_{-i}]
=
\E[1-X_i\mid\mathbf X_{-i}]
=
1-\mu_i
\quad\text{almost surely}.
\]
Thus \cref{thm:average-mean-sympol} applies to
$\mathbf X$ with conditional mean bounds
$(\mu_1,\ldots,\mu_n)$ and to $\mathbf Y$ with conditional mean
bounds $(1-\mu_1,\ldots,1-\mu_n)$. It follows that
\[
\mathbb P\left\{
L_{\mathrm{SP,av}}(\mathbf X)
\le\bar\mu\le
U_{\mathrm{SP,av}}(\mathbf X)
\right\}
\ge
1-\alpha_L-\alpha_U.
\]

Thus the SymPol endpoint formulas are unchanged: under the exact
conditional-mean condition
\eqref{eq:conditional-mean-independence}, they cover the average of
the component means even when those means differ. This condition is
weaker than independence and includes independent observations as a
special case.

Moreover, writing
\(
\bar x
=
n^{-1}\sum_{i=1}^n x_i,
\)
Maclaurin's inequality gives
\[
L_{\mathrm{SP,av}}(\mathbf x)
\le
\bar x
\le
U_{\mathrm{SP,av}}(\mathbf x),
\]
so the interval is nonempty and contains the sample average.

\section{A fast algorithm for SymPol}
\label{sec:comp}

We   give a fast algorithm  for computing $\max_{0\le k\le n} A_k(\mathbf E)$. Recall the polynomial identity
$$
P(z):=\prod_{i=1}^n (1+E_i z)
    =\sum_{k=0}^n S_k(\mathbf E) z^k.
$$
Therefore, the elementary symmetric polynomial $S_k(\mathbf E)$ is exactly the
coefficient of $z^k$ in $P(z)$, and
$$
A_k(\mathbf E)=\frac{S_k(\mathbf E)}{\binom{n}{k}},
\quad k=0,1,\dots,n.
$$
Hence, computing $\max_{0\le k\le n}A_k(\mathbf E)$ reduces to computing all
coefficients of the polynomial $P(z)$ and then taking the maximum after 
normalizing by $\binom{n}{k}$.
We recursively split the set of linear factors
$$
1+E_1 z,\dots,1+E_n z
$$
into two groups of nearly equal size, compute the product polynomial of each
group recursively, and then multiply the two resulting polynomials using a fast
polynomial multiplication routine.

Let $M(m)$ denote the complexity of multiplying two polynomials of degree at
most $m$. If $T(n)$ denotes the time required to compute
$
P(z),
$
then the divide-and-conquer algorithm satisfies the recurrence
$
T(n)=T(\lfloor n/2\rfloor)+T(\lceil n/2\rceil)+O(M(n)).
$
Thus,
$
T(n)=O(M(n)\log n).
$
In particular, if FFT-based multiplication is used, then
$
M(n)=O(n\log n),
$
and therefore
$$
T(n)=O(n\log^2 n).
$$
Once the coefficients of $P(z)$ are available, all values
$A_k(\mathbf E)$ can be obtained in an additional $O(n)$ time, and hence
$\max_{0\le k\le n} A_k(\mathbf E)$ is computed in overall time
$
O(M(n)\log n),
$
which becomes $O(n\log^2 n)$ under FFT-based multiplication.
We summarize the procedure in \cref{alg:fast-symmetric-max}.

\begin{algorithm}[t]
\caption{Fast computation of $\max_{0\le k\le n} A_k(\mathbf{E})$}
\label{alg:fast-symmetric-max}
\begin{algorithmic}[1]
\Require $E_1,\dots,E_n \ge 0$
\Ensure $\max_{0\le k\le n} A_k(\mathbf{E})$

\Function{BuildProduct}{$L,R$}
    \If{$L=R$}
        \State \Return $1+E_L z$
    \EndIf
    \State $M \gets \lfloor (L+R)/2 \rfloor$
    \State $P_{\mathrm{left}}(z)
        \gets \Call{BuildProduct}{L,M}$
    \State $P_{\mathrm{right}}(z)
        \gets \Call{BuildProduct}{M+1,R}$
    \State \Return
        $\Call{FFTMultiply}
        {P_{\mathrm{left}}(z),P_{\mathrm{right}}(z)}$
\EndFunction

\State $P(z)\gets\Call{BuildProduct}{1,n}$
\State Extract coefficients $c_0,\dots,c_n$ from
\Statex \hspace{\algorithmicindent}
    $\displaystyle P(z)=\sum_{k=0}^n c_k z^k$

\State $B_{\max}\gets 0$
\For{$k=0,\dots,n$}
    \State $A_k\gets c_k/\binom{n}{k}$
    \State $B_{\max}\gets\max\{B_{\max},A_k\}$
\EndFor
\State \Return $B_{\max}$
\end{algorithmic}
\end{algorithm}

\section{Conclusion} 

We establish applications of demi-supermartingales through a version of  Ville's inequality. This leads to new probabilistic inequalities, in particular, the validity of a new method, called SymPol, of combining (co-valid, compound) e-values into a p-value. SymPol is based on elementary symmetric polynomials, admits a fast algorithm, and is shown to yield a smaller (or equal) p-value than the KL-inf statistic does. Our results positively settle an explicit conjecture of~\cite{WZ03} and one of  \cite{gaffke2005three}.\footnote{\cite{gaffke2005three} has a stronger conjecture, which was recently proved by \cite{VlassisThomas2026}, four months after the first arXiv version of our paper; see \cite{MingEtAl2026Gaffke} for a comparative analysis. 
The assumptions and scope for our results differ from theirs, because we prove the inequalities in \cref{conjecture:gwz} under co-validity and heterogeneous means.} 
We hope that the  techniques of using demi-supermartingales with e-values will be useful in other problems in probability and statistics.

\subsection*{Acknowledgments}
AI tools were used to assist with some technical results and language editing. The authors maintain responsibility for the correctness of the claims. 
RW is supported by the Natural Sciences and Engineering Research Council of Canada (CRC-2022-00141, RGPIN-2024-03728). 
IW-S acknowledges support from the Miller Institute for Basic Research in Science.

\appendix

\section{Omitted proofs}

\subsection{Proof of \cref{proposition:sympol-ci}}
\label{proof:proposition:sympol-ci}
\begin{proof}[Proof of \cref{proposition:sympol-ci}]
Put $q=\alpha/2$. For $m\in(0,1]$, under the null $\mu\le m$,
the variables $E_i=X_i/m$ are independent e-variables. By homogeneity,
\[
  A_k(E_1,\ldots,E_n)=m^{-k}A_k(X_1,\ldots,X_n).
\]
Hence the SymPol test rejects $m$ exactly when
$m\le L_n^{\SP}$. At the true mean,
\[
\P\!\left(L_n^{\SP}\ge \mu\right)
=
\P\!\left(
  \max_{0\le k\le n}A_k(X_1/\mu,\ldots,X_n/\mu)
  \ge \frac{2}{\alpha}
\right)
\le \frac\alpha2,
\]
with the case $\mu=0$ handled directly because then $X_i=0$ almost surely.
Applying the same argument to $1-X_i$ gives
\[
  \P\!\left(U_n^{\SP}\le \mu\right)
  \le \frac\alpha2.
\]
The union bound proves coverage.

For nonemptiness, Maclaurin's inequalities imply
\[
  A_k(X_1,\ldots,X_n)^{1/k}\le \overline X_n.
\]
Since $q^{1/k}<1$, this yields
$L_n^{\SP}<\overline X_n$ whenever
$\overline X_n>0$. Applying the same argument to $1-X_i$ gives
$U_n^{\SP}>\overline X_n$ whenever
$\overline X_n<1$. If $\overline X_n=0$, then
$C_n^{\rm SP}
=[0,1-q^{1/n})$; if $\overline X_n=1$, then it equals
$(q^{1/n},1]$. Thus the  interval is always nonempty and contains the
sample mean.

Finally, \cref{theorem:maximal-inequality-co-valid} gives, for every $m$,
\[
W_n^\brackell(m)
\le \max_{0\le k\le n}A_k(X_1/m,\ldots,X_n/m),
\]
and analogously for the upper-tail statistics based on $1-X_i$.
Therefore the SymPol acceptance set is contained in the optimized-product
acceptance set, which yields the two outer endpoint inequalities.
\end{proof}

\subsection{Proof of \cref{thm:width}}
\label{proof:thm:width}
\begin{proof}[Proof of \cref{thm:width}]
Write
\(
    \overline X_n := \frac1n\sum_{i=1}^n X_i,
    \widehat\sigma_n^2 :=
    \frac1n\sum_{i=1}^n (X_i-\overline X_n)^2,
\)
and first consider the case that $\sigma > 0$, which implies $\mu\in(0,1)$.
By the strong law of large numbers,
\(
    \overline X_n\to\mu,
    \widehat\sigma_n^2\to\sigma^2
\)
almost surely. We work throughout on this probability-one event.

Set
\(
    c:=\log(2/\alpha)>0.
\)
We first show that $C_n$ lies inside an $O(n^{-1/2})$ neighborhood of
$\overline X_n$. Let $m\in[0,1]$ and put
\(
    d:=|m-\overline X_n|.
\)
Choose
\(
    \gamma
    :=
    \frac d2\,\operatorname{sgn}(\overline X_n-m).
\)
Then $|\gamma|\le1/2$, so $\gamma$ is feasible. Also
$|\gamma(X_i-m)|\le1/2$. Using the elementary inequality
\(
    \log(1+u)\ge u-u^2,
    \text{ for } |u|\le1/2,
\)
we get
\[
\begin{aligned}
    \frac1n\sum_{i=1}^n
    \log\bigl(1+\gamma(X_i-m)\bigr)
    &\ge
    \gamma(\overline X_n-m)
    -
    \gamma^2\frac1n\sum_{i=1}^n (X_i-m)^2  \\
    &\ge
    \frac{d^2}{2}-\frac{d^2}{4}
    =
    \frac{d^2}{4}.
\end{aligned}
\]
Therefore
\(
    \ell_n(m)\ge \frac{n d^2}{4}.
\)
Consequently, if $m\in C_n$, then $\ell_n(m)<c$, and hence
\(
    |m-\overline X_n|
    <
    \frac{2\sqrt c}{\sqrt n}.
\)
Thus $C_n$ is contained in a fixed $n^{-1/2}$-neighborhood of
$\overline X_n$.
Now define
\(
    r:=\sigma\sqrt{2c}
    =
    \sigma\sqrt{2\log(2/\alpha)}.
\)
We claim that the rescaled set
\[
    \sqrt n(C_n-\overline X_n)
    :=
    \left\{\sqrt n(m-\overline X_n):m\in C_n\right\}
\]
converges to $(-r,r)$ in the sense needed for widths.
Fix $\varepsilon\in(0,r)$. By the local quadratic approximation in \cref{lem:klinf-asymp},
uniformly for $|t|\le r-\varepsilon$,
\[
    \ell_n\left(\overline X_n+\frac{t}{\sqrt n}\right)
    \longrightarrow
    \frac{t^2}{2\sigma^2}
    \le
    \frac{(r-\varepsilon)^2}{2\sigma^2}
    <
    c.
\]
Hence, for all sufficiently large $n$,
\[
    \left[
        \overline X_n-\frac{r-\varepsilon}{\sqrt n},
        \overline X_n+\frac{r-\varepsilon}{\sqrt n}
    \right]
    \subseteq C_n.
\]
The interval above is contained in $[0,1]$ for all large $n$, since
$\overline X_n\to\mu\in(0,1)$.

For the reverse inclusion, the localization bound showed that every
$m\in C_n$ satisfies
\(
    \left|\sqrt n(m-\overline X_n)\right|
    <
    2\sqrt c.
\)
Applying the local quadratic approximation on the compact interval
$[-2\sqrt c,2\sqrt c]$, we get that uniformly over
$r+\varepsilon\le |t|\le 2\sqrt c$,
\[
    \ell_n\left(\overline X_n+\frac{t}{\sqrt n}\right)
    \longrightarrow
    \frac{t^2}{2\sigma^2}
    \ge
    \frac{(r+\varepsilon)^2}{2\sigma^2}
    >
    c.
\]
Therefore, for all sufficiently large $n$, no point of $C_n$ can satisfy
\(
    |m-\overline X_n|\ge \frac{r+\varepsilon}{\sqrt n}.
\)
Thus
\(
    C_n
    \subseteq
    \left[
        \overline X_n-\frac{r+\varepsilon}{\sqrt n},
        \overline X_n+\frac{r+\varepsilon}{\sqrt n}
    \right]
\)
eventually.
Combining the two inclusions, for all sufficiently large $n$,
\[
    \left[
        \overline X_n-\frac{r-\varepsilon}{\sqrt n},
        \overline X_n+\frac{r-\varepsilon}{\sqrt n}
    \right]
    \subseteq C_n
    \subseteq
    \left[
        \overline X_n-\frac{r+\varepsilon}{\sqrt n},
        \overline X_n+\frac{r+\varepsilon}{\sqrt n}
    \right].
\]
Taking widths gives
\(
    \frac{2(r-\varepsilon)}{\sqrt n}
    \le
    \width_n
    \le
    \frac{2(r+\varepsilon)}{\sqrt n}
\)
eventually. 
Since $\varepsilon>0$ was arbitrary,
\[
    \sqrt n\,\width_n\to 2r
    =
    2\sigma\sqrt{2\log(2/\alpha)}.
\]
This proves the case of $\sigma^2 > 0$.

We now handle degenerate distributions with zero variance.
Since $\sigma^2=0$, on a probability-one event,
\(
    X_i=\mu \text{ for every }i\ge1.
\)
We work on this event.
Fix $m\in[0,1]$. Then
\(
    \prod_{i=1}^n \bigl(1+\gamma(X_i-m)\bigr)
    =
    \bigl(1+\gamma(\mu-m)\bigr)^n.
\)
Therefore
\[
    \overline W_n(m)
    =
    \sup_{\gamma\in[-1/(1-m),\,1/m]}
    \bigl(1+\gamma(\mu-m)\bigr)^n.
\]
First suppose $m\in(0,1)$. If $m<\mu$, then $\mu-m>0$, so the quantity
\(
    1+\gamma(\mu-m)
\)
is increasing in $\gamma$. Hence the supremum is attained at the upper
endpoint $\gamma=1/m$, giving
\[
    \overline W_n(m)
    =
    \left(1+\frac{\mu-m}{m}\right)^n
    =
    \left(\frac{\mu}{m}\right)^n.
\]
If $m>\mu$, then $\mu-m<0$, so the same quantity is decreasing in $\gamma$.
The supremum is then attained at the lower endpoint
$\gamma=-1/(1-m)$, giving
\[
    \overline W_n(m)
    =
    \left(1+\frac{m-\mu}{1-m}\right)^n
    =
    \left(\frac{1-\mu}{1-m}\right)^n.
\]
Finally, if $m=\mu$, then
\(
    \overline W_n(\mu)=1.
\)
Let
\(
    c:=\log(2/\alpha).
\)
Since $2/\alpha>1$, we have $c>0$. For $m<\mu$, the condition
$m\in C_n$ is
\(
    \left(\frac{\mu}{m}\right)^n < e^c,
\)
or
\(
    m>\mu e^{-c/n}.
\)
Similarly, for $m>\mu$, the condition $m\in C_n$ is
\(
    \left(\frac{1-\mu}{1-m}\right)^n < e^c,
\)
which is equivalent to
\(
    m<1-(1-\mu)e^{-c/n}.
\)
Therefore, when $\mu\in(0,1)$,
\[
    C_n
    =
    \left(
        \mu e^{-c/n},
        1-(1-\mu)e^{-c/n}
    \right).
\]
Hence
\(
    \width_n=1-e^{-c/n},
\)
as claimed.

The endpoint cases are consistent with the same formula. If $\mu=0$, then
$X_i=0$ for every $i$, and
\(
    \overline W_n(0)=1.
\)
For $m>0$,
\(
    \overline W_n(m)
    =
    \left(\frac{1}{1-m}\right)^n.
\)
Thus
\(
    m\in C_n
    \Longleftrightarrow
    \left(\frac{1}{1-m}\right)^n<e^c,
\)
which is equivalent to
\(
    m<1-e^{-c/n}.
\)
Therefore
\(
    C_n=[0,1-e^{-c/n}),
\)
and again
\(
    \width_n=1-e^{-c/n}.
\)
Likewise, if $\mu=1$, then
\(
    C_n=(e^{-c/n},1],
\)
so again
\(
    \width_n=1-e^{-c/n},
\)
and thus
\(
    n\width_n\to c,
\)
as claimed.
\end{proof}

\subsection{Proof of \cref{lem:klinf-asymp}}
\label{proof:lem:klinf-asymp}
\begin{proof}[Proof of \cref{lem:klinf-asymp}]
    
 Fix $K<\infty$. For $|t|\le K$, define
\[
    m_{n,t}:=\overline X_n+\frac{t}{\sqrt n},
    \quad
    Y_{i,n,t}:=X_i-m_{n,t}.
\]
For all sufficiently large $n$, uniformly over $|t|\le K$, we have
$m_{n,t}\in(0,1)$, because $\overline X_n\to\mu\in(0,1)$.
Now
\[
    \frac1n\sum_{i=1}^n Y_{i,n,t}
    =
    \overline X_n-m_{n,t}
    =
    -\frac{t}{\sqrt n} =: a_{n,t}, 
\]
and
\[
    \frac1n\sum_{i=1}^n Y_{i,n,t}^2
    =
    \frac1n\sum_{i=1}^n
    \left(X_i-\overline X_n-\frac{t}{\sqrt n}\right)^2
    =
    \widehat\sigma_n^2+\frac{t^2}{n} =: b_{n,t}.
\]
Then, uniformly for $|t|\le K$,
\(
    a_{n,t}=O(n^{-1/2}),
    b_{n,t}\to\sigma^2.
\)
For $m=m_{n,t}$, define
\[
    h_{n,t}(\gamma)
    :=
    \frac1n\sum_{i=1}^n
    \log\bigl(1+\gamma Y_{i,n,t}\bigr).
\]
The function $h_{n,t}$ is concave in $\gamma$ on its feasible interval. We first
show that its maximizer is of order $n^{-1/2}$.
Choose some fixed $\delta\in(0,1/2)$. Since $|Y_{i,n,t}|\le1$, the interval
$[-\delta,\delta]$ lies inside the feasible interval for every $m\in[0,1]$.
For $\gamma\in(-\delta,\delta)$,
\[
    h_{n,t}'(\gamma)
    =
    \frac1n\sum_{i=1}^n
    \frac{Y_{i,n,t}}{1+\gamma Y_{i,n,t}}.
\]
Using
\[
    \frac{Y}{1+\gamma Y}
    =
    Y-\gamma\frac{Y^2}{1+\gamma Y},
\]
we get
\[
    h_{n,t}'(\delta)
    =
    a_{n,t}
    -
    \delta\frac1n\sum_{i=1}^n
    \frac{Y_{i,n,t}^2}{1+\delta Y_{i,n,t}},
\]
and
\[
    h_{n,t}'(-\delta)
    =
    a_{n,t}
    +
    \delta\frac1n\sum_{i=1}^n
    \frac{Y_{i,n,t}^2}{1-\delta Y_{i,n,t}}.
\]
Since $|Y_{i,n,t}|\le1$,
\[
    \frac1n\sum_{i=1}^n
    \frac{Y_{i,n,t}^2}{1+\delta Y_{i,n,t}}
    \ge
    \frac{b_{n,t}}{1+\delta},
\]
and similarly
\[
    \frac1n\sum_{i=1}^n
    \frac{Y_{i,n,t}^2}{1-\delta Y_{i,n,t}}
    \ge
    \frac{b_{n,t}}{1+\delta}.
\]
Because $b_{n,t}\to\sigma^2>0$ uniformly for $|t|\le K$ and
$a_{n,t}=O(n^{-1/2})$ uniformly for $|t|\le K$, it follows that, for all
large $n$,
\[
    h_{n,t}'(-\delta)>0
    \quad\text{and}\quad
    h_{n,t}'(\delta)<0
\]
uniformly over $|t|\le K$. Hence, by concavity, every maximizer
$\gamma_{n,t}^\star$ lies in $(-\delta,\delta)$.
At such a maximizer,
\(
    0=h_{n,t}'(\gamma_{n,t}^\star).
\)
Therefore
\[
    0
    =
    a_{n,t}
    -
    \gamma_{n,t}^\star
    \frac1n\sum_{i=1}^n
    \frac{Y_{i,n,t}^2}{1+\gamma_{n,t}^\star Y_{i,n,t}}.
\]
Since $|\gamma_{n,t}^\star|\le\delta$ and $|Y_{i,n,t}|\le1$,
\[
    \frac1n\sum_{i=1}^n
    \frac{Y_{i,n,t}^2}{1+\gamma_{n,t}^\star Y_{i,n,t}}
    \ge
    \frac{b_{n,t}}{1+\delta}.
\]
Thus, uniformly for $|t|\le K$,
\[
    |\gamma_{n,t}^\star|
    \le
    \frac{(1+\delta)|a_{n,t}|}{b_{n,t}}
    =
    O(n^{-1/2}).
\]
We now Taylor expand. Since $|Y_{i,n,t}|\le1$ and
$|\gamma_{n,t}^\star|=O(n^{-1/2})$, Taylor's formula gives, uniformly over
$|t|\le K$,
\[
    \log(1+\gamma Y)
    =
    \gamma Y-\frac{\gamma^2Y^2}{2}
    +
    O(|\gamma|^3),
\]
whenever $|\gamma|\le Cn^{-1/2}$ and $|Y|\le1$. Hence
\[
    h_{n,t}(\gamma)
    =
    \gamma a_{n,t}
    -
    \frac{\gamma^2 b_{n,t}}{2}
    +
    O(|\gamma|^3)
\]
uniformly for $|\gamma|\le Cn^{-1/2}$ and $|t|\le K$.
The quadratic function
\(
    q_{n,t}(\gamma)
    :=
    \gamma a_{n,t}
    -
    \frac{\gamma^2 b_{n,t}}{2}
\)
has maximum
\(
    \sup_\gamma q_{n,t}(\gamma)
    =
    \frac{a_{n,t}^2}{2b_{n,t}},
\)
attained at $\gamma=a_{n,t}/b_{n,t}=O(n^{-1/2})$. Therefore,
using the Taylor expansion both at the true maximizer
$\gamma_{n,t}^\star$ and at $a_{n,t}/b_{n,t}$, we obtain
\[
    \sup_\gamma h_{n,t}(\gamma)
    =
    \frac{a_{n,t}^2}{2b_{n,t}}
    +
    O(n^{-3/2})
\]
uniformly for $|t|\le K$. Multiplying by $n$ gives
\[
\begin{aligned}
    \ell_n\left(\overline X_n+\frac{t}{\sqrt n}\right)
    &=
    n\sup_\gamma h_{n,t}(\gamma)
    =
    \frac{n a_{n,t}^2}{2b_{n,t}}
    +
    O(n^{-1/2}) \\
 &  =
    \frac{t^2}{2(\widehat\sigma_n^2+t^2/n)}
    +
    O(n^{-1/2}).
    \end{aligned}
\]
Thus, uniformly for $|t|\le K$,
\[
    \ell_n\left(\overline X_n+\frac{t}{\sqrt n}\right)
    \longrightarrow
    \frac{t^2}{2\sigma^2}.
\]
This is the local quadratic approximation.
\end{proof}

\subsection{Proof of \cref{th:KL-inf}}\label{proof:th:KL-inf}
\begin{proof}[Proof of \cref{th:KL-inf}]
It is enough to prove the normalized identity with \(\widehat \Q_n:=\widehat \Q_n(\mathbf x)
=\frac1n\sum_{i=1}^n\delta_{x_i}\). For
\(a\in[0,1]\), set
\[
  f_a(x):=1-a+ax=1+a(x-1),
  \quad
  \phi(a):=\E_{\widehat \Q_n}[\log f_a(X)].
\]
The function \(\phi\) is concave on \([0,1]\), with values in
\([-\infty,\infty)\). It attains its supremum: if \(\widehat \Q_n(\{0\})=0\), then
\(\phi\) is continuous on \([0,1]\); if \(\widehat \Q_n(\{0\})>0\), then
\(\phi(a)\to-\infty\) as \(a\uparrow1\), while \(\phi(0)=0\), so the maximum
is attained in some compact subinterval of \([0,1)\).

\paragraph{Weak duality.}
Fix \(\Q\in\mathcal C\) and \(a\in[0,1]\). Since \(f_a\ge0\),
\[
  \E_\Q [f_a(X)]
  = 1-a+a\E_\Q [X]
  \le 1.
\]
The standard variational inequality
\[
  \KL(\widehat \Q_n\|\Q)
  \ge
  \E_{\widehat \Q_n}\left[\log f_a(X)\right]-\log \E_\Q [f_a(X)]
\]
holds for every nonnegative \(f_a\), with the usual extended-value
conventions. For completeness, when \(0<\E_\Q [f_a]<\infty\), it follows by
defining \(\d \Q^{(a)}=f_a\d \Q/\E_\Q [f_a]\) and using
\(\KL(\widehat \Q_n\|\Q^{(a)})\ge0\); the remaining cases follow by the same inequality
with extended values. Therefore
\(
  \KL(\widehat \Q_n\|\Q)
  \ge
  \E_{\widehat \Q_n}[\log f_a(X)],
\)
and hence
\[
  \inf_{\Q\in\mathcal C}\KL(\widehat \Q_n\|\Q)
  \ge
  \sup_{a\in[0,1]}\phi(a).
\]

\paragraph{Equality and the attaining distribution.}
Let \(a^\star\in[0,1]\) maximize \(\phi\). We construct
\(\Q^\star\in\mathcal C\) such that
\(
  \KL(\widehat \Q_n\|\Q^\star)=\phi(a^\star).
\)
This proves the reverse inequality.

\smallskip
\noindent\textbf{Case 1: \(a^\star=0\).}
Since \(\phi\) is concave and is maximized at the left endpoint,
\[
  \phi'_+(0)=\E_{\widehat \Q_n}[X-1]\le0.
\]
Thus \(\E_{\widehat \Q_n} [X]\le1\), so \(\widehat \Q_n\in\mathcal C\). Taking \(\Q^\star=\widehat \Q_n\) gives
\[
  \KL(\widehat \Q_n\|\Q^\star)=0=\phi(0).
\]

\smallskip
\noindent\textbf{Case 2: \(0<a^\star<1\).}
The first-order condition is
\[
  0=\phi'(a^\star)
  =
  \E_{\widehat \Q_n}\left[\frac{X-1}{1-a^\star+a^\star X}\right].
\]
Define \(\Q^\star\) on the support of \(\widehat \Q_n\) by
\[
  \frac{\d\Q^\star}{\d\widehat \Q_n}(x)
  =
  \frac{1}{1-a^\star+a^\star x}.
\]
This is a probability measure. Indeed,
\[
  \E_{\widehat \Q_n}\left[\frac{1}{1-a^\star+a^\star X}\right]
  =
  \E_{\widehat \Q_n}\left[
    1
    -
    a^\star
    \frac{X-1}{1-a^\star+a^\star X}
  \right]
  =
  1.
\]
It also has mean one:
\[
  \E_{\Q^\star}[X]
  =
  \E_{\widehat \Q_n}\left[\frac{X}{1-a^\star+a^\star X}\right]
  =
  \E_{\widehat \Q_n}\left[\frac{X-1}{1-a^\star+a^\star X}\right]
  +
  \E_{\widehat \Q_n}\left[\frac{1}{1-a^\star+a^\star X}\right]
  =
  1.
\]
Thus \(\Q^\star\in\mathcal C\). Finally,
\[
  \KL(\widehat \Q_n\|\Q^\star)
  =
  \E_{\widehat \Q_n}\left[\log\!\left(\frac{\d\widehat \Q_n}{\d\Q^\star}\right)\right]
  =
  \E_{\widehat \Q_n}\left[\log(1-a^\star+a^\star X)\right]
  =
  \phi(a^\star).
\]

\smallskip
\noindent\textbf{Case 3: \(a^\star=1\).}
This case can occur only when \(\widehat \Q_n(\{0\})=0\), since otherwise
\(\phi(1)=-\infty<\phi(0)=0\). Concavity and optimality at the right endpoint
give
\[
  \phi'_-(1)
  =
  \E_{\widehat \Q_n}\left[\frac{X-1}{X}\right]
  =
  1-\E_{\widehat \Q_n}\left[\frac1X\right]
  \ge0.
\]
Hence
\[
  Z:=\E_{\widehat \Q_n}\left[\frac1X\right]\le1.
\]
Define \(\Q^\star\) by putting mass \(\d \widehat \Q_n(x)/x\) on the positive support of
\(\widehat \Q_n\), and putting the remaining mass \(1-Z\) at zero:
\[
  \Q^\star(A)
  :=
  \int_{A\cap(0,\infty)} \frac1x\d \widehat \Q_n(x)
  +
  (1-Z)\mathbf 1_{\{0\in A\}}.
\]
Then \(\Q^\star\) is a probability measure and
\[
  \E_{\Q^\star}[X]
  =
  \int_{(0,\infty)} x\frac1x\d \widehat \Q_n(x)
  =
  1,
\]
so \(\Q^\star\in\mathcal C\). On the support of \(\widehat \Q_n\), \(\d\Q^\star/\d\widehat \Q_n=1/X\),
and therefore
\[
  \KL(\widehat \Q_n\|\Q^\star)
  =
  \E_{\widehat \Q_n}[\log X]
  =
  \phi(1).
\]

In all cases, the lower bound from weak duality is attained. Therefore
\[
  \inf_{\Q\in\mathcal C}\KL(\widehat \Q_n\|\Q)
  =
  \sup_{a\in[0,1]}\E_{\widehat \Q_n}[\log(1-a+aX)],
\]
this is exactly the stated identity.
\end{proof}

\subsection{Limiting width of the SymPol interval}\label{app:sympol-size}

Write
\[
 \overline X_n=\frac1n\sum_{i=1}^nX_i,
 \quad
 \widehat\sigma_n^2
 =
 \frac1n\sum_{i=1}^n(X_i-\overline X_n)^2.
\]

\begin{theorem}[Limiting width of the SymPol interval]
\label{thm:sympol-width}
Let $X_1,X_2,\ldots$ be iid random variables on 
$[0,1]$, and write
\(
  \mu=\mathbb E [X_1],
  \sigma^2=\operatorname{Var}(X_1).
\)
Fix $\alpha\in(0,1)$ and set
\(
 q:=\frac{\alpha}{2},
 a:=\log\frac1q.
\)
Let $L_n^{\SP}$ and $U_n^{\SP}$
be given in \cref{proposition:sympol-ci}.
If $\sigma>0$, then, almost surely,
\begin{align}
 L_n^{\SP}
 &=
 \overline X_n-
 \widehat\sigma_n
 \sqrt{\frac{2\log(2/\alpha)}{n}}
 +o(n^{-1/2}),
 \label{eq:sympol-lower-expansion}\\
 U_n^{\SP}
 &=
 \overline X_n+
 \widehat\sigma_n
 \sqrt{\frac{2\log(2/\alpha)}{n}}
 +o(n^{-1/2}),
 \label{eq:sympol-upper-expansion}\\
  \sqrt n\,
 \bigl(U_n^{\SP}&-L_n^{\SP}\bigr)
 \longrightarrow
 2\sigma\sqrt{2\log\frac{2}{\alpha}}.
\end{align}
If $\sigma=0$, then $\bigl(U_n^{\SP}-L_n^{\SP}\bigr) = 1 - (\alpha/2)^{1/n}$ and so $n\,
 \bigl(U_n^{\SP}-L_n^{\SP}\bigr)
 \longrightarrow
 \log\frac{2}{\alpha}.$
\end{theorem}

\begin{proof}
We divide the proof into four steps.
\paragraph{Step 1: inversion of the SymPol test.}

The SymPol inequality says that, for co-valid e-variables
$E_1,\ldots,E_n$ and every $q\in(0,1)$,
\begin{equation}
 \mathbb P\left(
   \max_{0\le k\le n}
   A_k(E_1,\ldots,E_n)
   \ge \frac1q
 \right)
 \le q.
 \tag{S}
\end{equation}

For testing
\(
 H_{0,L}(\theta):\mu\le\theta,
\)
use
\(
 E_i=\frac{X_i}{\theta}.
\)
Under the null, these are independent e-variables, and homogeneity
gives
\[
 A_k
 \left(
 \frac{X_1}{\theta},\ldots,\frac{X_n}{\theta}
 \right)
 =
 \frac{A_k(X_1,\ldots,X_n)}{\theta^k}.
\]
Thus the nonrejected values of $\theta$ form the ray
$(L_n^{\SP},1]$, where
\[
 L_n^{\SP}
 =
 \max_{1\le k\le n}
 \left\{
 qA_k(X_1,\ldots,X_n)
 \right\}^{1/k}.
\]
For the other direction, test
\(
 H_{0,U}(\theta):\mu\ge\theta
\)
using
\(
 \widetilde E_i=\frac{1-X_i}{1-\theta}.
\)
This gives
\[
 U_n^{\SP}
 =
 1-
 \max_{1\le k\le n}
 \left\{
 qA_k(1-X_1,\ldots,1-X_n)
 \right\}^{1/k}.
\]
Using error $q=\alpha/2$ in each tail gives coverage at least
$1-\alpha$ by the union bound.

\paragraph{Step 2: a no-collision estimate.}

We first establish an auxiliary birthday-problem estimate. Let
\(
 \mathbf p^{(n)}=(p_{1,n},\ldots,p_{n,n})
\)
be a probability vector satisfying
\(
 \max_i p_{i,n}\le\frac{C}{n}
\)
for some fixed $C<\infty$. Let $I_1,\ldots,I_k$ be iid with
law $\mathbf p^{(n)}$, and let $\mathcal D_k$ denote the event that all
$I_j$ are distinct.

For every fixed $0<\delta<M<\infty$, uniformly over
\(
 \delta\sqrt n\le k\le M\sqrt n,
\)
we claim that
\begin{equation}
 \log\mathbb P(\mathcal D_k)
 =
 -\binom{k}{2}\sum_{i=1}^n p_{i,n}^2+o(1).
 \label{eq:no-collision}
\end{equation}

To prove this, let
\[
 W:=
 \sum_{1\le r<s\le k}
 \ind\{I_r=I_s\}
\]
be the number of colliding pairs, and define
\[
 s_{j,n}:=\sum_{i=1}^n p_{i,n}^j,
 \quad
 \lambda_n:=\binom{k}{2}s_{2,n}.
\]
The assumptions imply $\lambda_n=O(1)$.
We show that, for every fixed positive integer $r$,
\begin{equation}
 \mathbb E[(W)_r]=\lambda_n^r+o(1),
 \label{eq:factorial-moments}
\end{equation}
where
\[
 (W)_r=W(W-1)\cdots(W-r+1).
\]

Expand $(W)_r$ as a sum over ordered $r$-tuples of distinct
edges of the complete graph on the $k$ draw positions. If those
$r$ edges are vertex-disjoint, their joint probability is
$s_{2,n}^r$. The total matching contribution is therefore
\[
 \frac{(k)_{2r}}{2^r}s_{2,n}^r
 =
 \lambda_n^r+o(1).
\]

Now consider a nonmatching edge tuple. Let $v$ be the number of
vertices used by its graph and let $c$ be the number of connected
components. If the component sizes are $v_1,\ldots,v_c$, then the
probability of the corresponding intersection of collision events is
\(
 \prod_{j=1}^c s_{v_j,n}.
\)
For every $m\ge2$,
\[
 s_{m,n}
 \le
 \left(\max_i p_{i,n}\right)^{m-1}
 \sum_i p_{i,n}
 =
 O\left(n^{-(m-1)}\right).
\]
It follows that the intersection probability is
\(
 O\left(n^{-(v-c)}\right).
\)
There are $O(k^v)=O(n^{v/2})$ embeddings of any fixed graph
type. Since a nonmatching graph has at least one component with at
least three vertices,
\(
 v\ge2c+1.
\)
Its total contribution is consequently
\[
 O\left(n^{v/2-(v-c)}\right)
 =
 O\left(n^{c-v/2}\right)
 =
 O(n^{-1/2}).
\]
There are only finitely many graph types for fixed $r$, which proves
\eqref{eq:factorial-moments}.

Take an arbitrary subsequence along which
$\lambda_n\to\lambda$. The Bonferroni inequalities give, for each
fixed integer $R$,
\[
 \sum_{r=0}^{2R+1}
 \frac{(-1)^r}{r!}\mathbb E[(W)_r]
 \le
 \mathbb P(W=0)
 \le
 \sum_{r=0}^{2R}
 \frac{(-1)^r}{r!}\mathbb E[(W)_r].
\]
Using \eqref{eq:factorial-moments}, then letting $n\to\infty$ and
subsequently $R\to\infty$, gives
\[
 \mathbb P(W=0)=e^{-\lambda_n}+o(1).
\]
The argument applies to every subsequence, and therefore the error is
uniform for $k/\sqrt n\in[\delta,M]$. Since $\lambda_n$ is
uniformly bounded, taking logarithms proves
\eqref{eq:no-collision}.

\paragraph{Step 3: asymptotics of the optimized elementary symmetric mean.}

We now prove a deterministic lemma. Let
\(
\mathbf x_n
=
(x_{1,n},\ldots,x_{n,n})
\in[0,B]^n
\)
for a fixed $B<\infty$, and define
\[
 m_n:=\frac1n\sum_{i=1}^n x_{i,n},
 \quad
 s_n^2:=
 \frac1n\sum_{i=1}^n(x_{i,n}-m_n)^2.
\]
Assume
\(
 m_n\to m>0,
 s_n\to s>0.
\)
Put
\[
 B_n(q):=
 \max_{1\le k\le n}
 \left\{
 qA_k(x_{1,n},\ldots,x_{n,n})
 \right\}^{1/k}.
\]
We claim that
\begin{equation}
 \sqrt n\,\bigl(m_n-B_n(q)\bigr)
 \longrightarrow
 s\sqrt{2\log(1/q)}.
 \label{eq:deterministic-claim}
\end{equation}
Normalize the observations by setting
\[
 z_{i,n}:=\frac{x_{i,n}}{m_n},
 \quad
 p_{i,n}:=\frac{z_{i,n}}{n},
 \quad
 c_n^2:=\frac{s_n^2}{m_n^2}.
\]
Then
\(
 \sum_{i=1}^n p_{i,n}=1,
\)
and, eventually,
\(
 \max_i p_{i,n}\le\frac{C}{n}
\)
for some fixed $C$.

Let $I_1,\ldots,I_k$ be sampled iid from $\mathbf p^{(n)}$ and
let $J_1,\ldots,J_k$ be sampled iid uniformly from
$\{1,\ldots,n\}$. Direct expansion gives the exact identity
\begin{equation}
 A_k(z_{1,n},\ldots,z_{n,n})
 =
 \frac{
 \mathbb P(I_1,\ldots,I_k\text{ are distinct})
 }{
 \mathbb P(J_1,\ldots,J_k\text{ are distinct})
 }.
 \label{eq:birthday-ratio}
\end{equation}
Indeed,
\[
 \mathbb P(I_1,\ldots,I_k\text{ distinct})
 =
 \frac{k!}{n^k}
 \sum_{|S|=k}\prod_{i\in S}z_{i,n},
\]
whereas
\[
 \mathbb P(J_1,\ldots,J_k\text{ distinct})
 =
 \frac{k!}{n^k}\binom nk.
\]

Moreover,
\[
 \sum_{i=1}^n p_{i,n}^2
 =
 \frac{1+c_n^2}{n}.
\]
Applying \eqref{eq:no-collision} to the numerator and denominator of
\eqref{eq:birthday-ratio} gives, uniformly for
$\delta\sqrt n\le k\le M\sqrt n$,
\begin{equation}
 \log A_k(x_{1,n},\ldots,x_{n,n})
 =
 k\log m_n
 -
 \frac{c_n^2}{n}\binom{k}{2}
 +o(1).
 \label{eq:Ak-expansion}
\end{equation}

Let
\(
 a:=\log\frac1q
\)
and write
\[
 \rho_{k,n}
 :=
 \frac{
 \{qA_k(x_{1,n},\ldots,x_{n,n})\}^{1/k}
 }{m_n}.
\]
For $k/\sqrt n$ in a fixed compact subset of $(0,\infty)$,
\eqref{eq:Ak-expansion} yields, uniformly,
\begin{equation}
 \sqrt n\log\rho_{k,n}
 =
 -\frac{a\sqrt n}{k}
 -
 \frac{c_n^2(k-1)}{2\sqrt n}
 +o(1).
 \label{eq:rho-expansion}
\end{equation}

If
\(
 \frac{k}{\sqrt n}\to\kappa\in(0,\infty),
\)
then the right-hand side converges to
\[
 -f(\kappa),
 \quad
 f(\kappa):=
 \frac{a}{\kappa}+\frac{c^2\kappa}{2},
 \quad
 c:=\frac{s}{m}.
\]
The function $f$ has the unique minimizer
\(
 \kappa_*=\frac{\sqrt{2a}}{c}
\)
and
\(
 f(\kappa_*)=c\sqrt{2a}.
\)
It remains to show that degrees far from order $\sqrt n$ cannot
maximize. Newton's inequalities imply Maclaurin's inequalities:
\begin{equation}
 A_k(\mathbf x_n)^{1/k}
 \text{ is nonincreasing in }k,
 \quad
 A_k(\mathbf x_n)^{1/k}
 \le A_{1}(\mathbf x_n)=m_n.
 \tag{M}
\end{equation}
Choose
\(
 0<\delta<\kappa_*<M
\)
such that
\(
 \frac{a}{\delta}>c\sqrt{2a},
 \frac{c^2M}{2}>c\sqrt{2a}.
\)
For $k\le\delta\sqrt n$, equation (M) gives
\[
 \sqrt n\log\rho_{k,n}
 \le
 -\frac{a\sqrt n}{k}
 \le
 -\frac{a}{\delta}.
\]
Thus such small degrees cannot attain the asymptotic optimum.
For $k\ge M\sqrt n$, let
\(
 k_M:=\lfloor M\sqrt n\rfloor.
\)
Again by equation (M), and because $q^{1/k}\le1$,
\[
 \rho_{k,n}
 \le
 \frac{A_{k_M}(\mathbf x_n)^{1/k_M}}{m_n}.
\]
Using \eqref{eq:Ak-expansion} at $k_M$ gives
\[
 \limsup_{n\to\infty}
 \sup_{k\ge M\sqrt n}
 \sqrt n\log\rho_{k,n}
 \le
 -\frac{c^2M}{2}.
\]
Thus degrees larger than $M\sqrt n$ also cannot attain the
asymptotic optimum.

It follows that the maximizing degree lies in
$[\delta\sqrt n,M\sqrt n]$ asymptotically. Taking the maximum in
\eqref{eq:rho-expansion} therefore yields
\[
 \sqrt n
 \log\frac{B_n(q)}{m_n}
 \longrightarrow
 -c\sqrt{2a}.
\]
Since the logarithm is $O(n^{-1/2})$, exponentiation gives
\[
 B_n(q)
 =
 m_n-
 \frac{s\sqrt{2a}}{\sqrt n}
 +o(n^{-1/2}).
\]
This proves \eqref{eq:deterministic-claim}. Since $s_n\to s$, it
may equivalently be written as
\begin{equation}
 B_n(q)
 =
 m_n-
 s_n\sqrt{\frac{2\log(1/q)}{n}}
 +o(n^{-1/2}).
 \label{eq:Bn-expansion}
\end{equation}

\paragraph{Step 4: application to the confidence endpoints.}

Assume first that $\sigma>0$. Since $X_i\in[0,1]$, this implies
\(
 0<\mu<1.
\)
By the strong law of large numbers,
\(
 \overline X_n\to\mu,
 \widehat\sigma_n\to\sigma
 \text{ almost surely}.
\)
Apply \eqref{eq:Bn-expansion} pathwise to $x_{i,n}=X_i$. Since
$q=\alpha/2$, this gives
\[
 L_n^{\SP}
 =
 \overline X_n-
 \widehat\sigma_n
 \sqrt{\frac{2\log(2/\alpha)}{n}}
 +o(n^{-1/2})
 \quad\text{almost surely}.
\]

Apply the same result to $1-X_i$. Its empirical mean is
$1-\overline X_n$ and its empirical variance is again
$\widehat\sigma_n^2$. Therefore,
\[
 1-U_n^{\SP}
 =
 (1-\overline X_n)-
 \widehat\sigma_n
 \sqrt{\frac{2\log(2/\alpha)}{n}}
 +o(n^{-1/2})
 \quad\text{almost surely}.
\]
Equivalently,
\[
 U_n^{\SP}
 =
 \overline X_n+
 \widehat\sigma_n
 \sqrt{\frac{2\log(2/\alpha)}{n}}
 +o(n^{-1/2})
 \quad\text{almost surely}.
\]

Subtracting the endpoint expansions gives
\[
 U_n^{\SP}-L_n^{\SP}
 =
 2\widehat\sigma_n
 \sqrt{\frac{2\log(2/\alpha)}{n}}
 +o(n^{-1/2}).
\]
Since $\widehat\sigma_n\to\sigma$ almost surely, it follows that
\[
 \sqrt n\,
 \bigl(U_n^{\SP}-L_n^{\SP}\bigr)
 \longrightarrow
 2\sigma\sqrt{2\log\frac{2}{\alpha}}
 \quad\text{almost surely}.
\]

Finally, suppose that $\sigma=0$. Then $X_i=\mu$ almost surely, so
\[
 A_k(\mathbf X)=\mu^k,
 \quad
 A_k(\mathbf 1-\mathbf X)=(1-\mu)^k.
\]
Since $q^{1/k}$ is increasing in $k$,
\(
 L_n^{\SP}=\mu q^{1/n},\ 
 U_n^{\SP}=1-(1-\mu)q^{1/n}.
\)
Consequently,
\[
 U_n^{\SP}-L_n^{\SP}
 =
 1-q^{1/n}
 =
 \frac{\log(1/q)}{n}+o(n^{-1}),
\]
as desired. 
\end{proof}

\bibliographystyle{plainnat}
\bibliography{bib}

\end{document}